\documentclass[journal]{IEEEtran}
\IEEEoverridecommandlockouts
\usepackage{comment}

\usepackage{cite}
\usepackage{overpic}
\usepackage{amsmath,amssymb,amsfonts}
\usepackage{textcomp}
\usepackage{xcolor}
\usepackage{float}
\usepackage{amsthm}
\usepackage{graphicx}
\usepackage{epstopdf}
\usepackage{subfigure}
\usepackage{adjustbox}
\usepackage{stfloats}
\usepackage{amsmath,bm}
\usepackage{amsfonts}
\usepackage{amssymb}
\usepackage{color}
\usepackage{multirow}
\usepackage{multicol}
\usepackage{soul,xcolor}
\usepackage{algorithm}
\usepackage{algorithmic}
\usepackage{url}

\newcommand{\mathacr}[1]{\mathsf{#1}}
\theoremstyle{plain}

\newtheorem{lemma}{Lemma}

\newcommand{\vect}[1]{\mathbf{#1}}

\def\Htran{\mbox{\tiny $\mathrm{H}$}}
\def\Ttran{\mbox{\tiny $\mathrm{T}$}}
\def\CN{\mathcal{N}_{\mathbb{C}}} 

\begin{document}

\title{Unlocking the Energy-Saving Potential in O-RAN Cell-Free Massive MIMO by Joint Orchestration of Radio, Wireless Fronthaul, and Cloud Resources}

\author{Ozan Alp Topal, Özlem Tuğfe Demir, Emil Björnson, and Cicek Cavdar

\thanks{O. A. Topal, E. Björnson and  C. Cavdar are with the School of Electrical Engineering and Computer Science, KTH Royal Institute of Technology, Stockholm, Sweden (e-mail: \{oatopal, emilbjo, cavdar\}@kth.se). Ö. T. Demir was with the Electrical and Electronics Engineering at TOBB University of Economics and Technology, Ankara, Turkiye. Now, she is with the Electrical and Electronics Engineering at Bilkent University, Ankara, Turkiye (ozlemtugfedemir@bilkent.edu.tr).  }  
\thanks{ This work has been part of Celtic-Next project RAI-6Green: Robust and AI Native 6G for Green Networks with project-id: C2023/1-9, partly supported by the Swedish funding agency Vinnova. The work by \"O. T. Demir was supported by the 2232-B International Fellowship for Early Stage Researchers Programme and by the TÜBİTAK Project 122C149 (Intelligent End-to-End Design of Energy-Efficient and Hardware Impairments-Aware Cell-Free Massive MIMO for Beyond 5G), both funded by the Scientific and Technological Research Council of Turkiye (TUBITAK). E.~Bj\"ornson was supported by the FFL18-0277 and SUCCESS grants from SSF.}

}

\maketitle

\begin{abstract}
Network virtualization and cloudification in Open Radio Access Networks (O-RAN) enable joint orchestration of the processing and fronthaul resources, which are essential for realizing the energy-saving potential of cell-free massive MIMO networks. To harness this potential, we investigate cell-free massive MIMO deployed over an O-RAN architecture with a wireless fronthaul that removes the need for fiber deployment. We first model the end-to-end power consumption under wireless fronthaul. Then, we propose a joint orchestration framework for radio, fronthaul, and processing resources that minimizes end-to-end power consumption while satisfying user-equipment (UE) rate requirements and wireless-fronthaul constraints. Two algorithms are developed: a scenario-sampling/group-Lasso method for centralized precoding and a block-coordinate descent method for distributed precoding. Numerical results show that centralized precoding significantly outperforms distributed precoding. End-to-end resource orchestration provides up to $70\%$ energy-savings compared to cloud-only orchestration and up to $15\%$ compared to radio-only orchestration. Moreover, distributing the same total number of antennas across the coverage area, rather than concentrating them at a few radio units (RUs), substantially reduces network power consumption, demonstrating that cell-free massive MIMO can deliver both high performance and high energy efficiency in future mobile networks.
\end{abstract}

\vspace{-1mm}


\vspace{-1mm}

\section{Introduction}
Cell-free massive MIMO (multiple-input multiple-output) is a promising candidate for future mobile networks thanks to the improved fairness among user equipment (UEs). While modern cellular networks rely on the centralized deployment of a large number of antennas using massive MIMO technology, cell-free massive MIMO is based on densely deploying many cooperating radio units (RUs) with a smaller number of antennas that are capable of coherent joint transmission/reception in a region \cite{cfmMIMOOr}. 
Given the distributed nature of cell-free massive MIMO, more radio equipment, fronthaul, and processing resources can be deployed, which, without proper control mechanisms, risks significantly increasing the total power consumption \cite{ericsson_2022_energy_5g}. 

To address this, a significant amount of literature has been devoted to developing and investigating energy efficiency in cell-free massive MIMO networks. In \cite{GroupSparsePrecGreen}, the authors propose joint RU and fronthaul activation/deactivation to minimize the network energy consumption. However, their method requires excessive control (in each coherence block). \cite{EEcfmMIMOSparsePrec} overcomes this by proposing a sparse, large-scale processing-based energy efficiency maximization method, but their model neglects the impact of the number of antennas in the processing and fronthaul power. \cite{CellFreeAntennaOptimization} focuses only on radio energy-efficiency maximization, and \cite{minimize_energy_cf}, \cite{radio_min_sparse} focus on radio power minimization through RU shutdown, ignoring the effect of fronthaul and processing. 
Network power consumption should consider the power consumption of radio, fronthaul, and cloud processing elements as in  \cite{ozlem_jsac}, where joint optimization of these elements presents up to $30\%$ energy-saving compared to the case where radio resources are optimized independently of the cloud. 

\vspace{-3mm}
\subsection{Cell-Free Massive MIMO in O-RAN}
Another limitation of the listed works (except \cite{ozlem_jsac}) is the limited consideration of the radio access network (RAN) architecture and the impact of the chosen functional split. Functional splitting depends on separating network functions between the centralized (or distributed) cloud processing and RUs. Centralizing processing in all-purpose cloud centers makes processing in the network significantly more energy-efficient, a concept that has been studied in centralized-RAN (C-RAN), virtualized cloud-RAN \cite{cloudRAN}, and more recently in open RAN (O-RAN) literature \cite{ORAN}. The specific names of the chosen splits vary depending on the different RAN technologies. In this paper, we follow the naming conventions outlined in the 3GPP standardization \cite{3gpp_tr_38_816}. Considering the O-RAN architecture, in this paper, O-Cloud refers to the virtualized all-purpose processor, RU refers to the radio unit, and fronthaul is the link between these two units. To enable coherent joint transmission (CJT) in the cell-free massive MIMO, the only splits that can be implemented on the fronthaul are inter-physical layer (PHY) splits, specifically split options $8$, $7.1$, and $7.2$. For the higher splits (such as Option $6$), the tight synchronization between RUs cannot be guaranteed; therefore, CJT is not possible \cite{larsen_survey_2019}. Among the possible options, Option $8$ centralizes all processing in the cloud unit, converting RUs into solely RF components. In Option $7.1$, processing up to and including discrete Fourier transform/ inverse discrete Fourier transform (DFT/IDFT) is performed at the RU-site, enforcing the deployment of processors to the RU-site, which increases the total power consumption, but reduces the fronthaul load. In Option $7.2$, the processing up to and including precoding is done at the RU-site, further increasing the total power consumption, and reducing the fronthaul load. Although \cite{ozlem_jsac} compares Option $8$ and $7.1$ from the cell-free massive MIMO perspective, it does not utilize the advantage of centralized precoding in these options. As demonstrated in \cite{cell-free-book}, centralized precoding can provide significant performance improvement by allowing the cloud unit to decide the precoding based on channel observations from all the RUs. 

\vspace{-4mm}
\subsection{Wireless Fronthaul in  Cell-Free Massive MIMO}
\vspace{-1mm}

The connection to the centralized cloud is assumed to be an optical transport network in the mentioned prior works, requiring expensive fiber cable deployment for all RUs \cite{cell-free-book}. While this way of deployment for cell-free massive MIMO networks provides a robust transport network, the network deployment has two significant shortcomings. First, the deployment cost becomes considerably higher, and second, it limits the potential realization of cell-free networks in geographical areas that lack access to fiber cables. Wireless fronthaul has been investigated for small cells \cite{8010338, milano}, and massive MIMO \cite{7306533, milano}. Currently, some products enable wireless fronthaul for remote radio equipment by utilizing point-to-point links \cite{ericsson_minilink_6352}. The performance of wireless fronthaul for cell-free networks has been analyzed in \cite{umurhan, wireless_f}; however, these studies do not consider the rate requirements associated with different functional splits, and the energy-efficient operation has not been addressed.  In \cite{neetu}, an optimal joint RU activation and power allocation algorithm is proposed to minimize the network power consumption for unmanned aerial vehicle (UAV) networks with wireless fronthaul limitations. However, the energy reduction is limited due to only allowing RU shutdown, similar to \cite{ozlem_jsac}. In the conference version of this paper \cite{AsilomarConf24}, we proposed activating/deactivating each antenna element jointly with fronthaul and access power allocation instead of completely turning the RU on and off, which has allowed $40\%$ higher energy savings thanks to more refined control considering distributed precoding.

\vspace{-4mm}
\subsection{Contributions}

In this work, we propose a joint processing, fronthaul, and radio resource orchestration to achieve an end-to-end energy-efficient cell-free massive MIMO network with wireless fronthaul. More specifically, we jointly optimize the active processors, the number of antennas, time, and power allocation for the mmWave fronthaul at the open cloud (O-Cloud); the mmWave receiver for the fronthaul, transmit power, number of antenna elements, and active processors at the radio site to minimize end-to-end power consumption while guaranteeing spectral efficiency (SE) requirements of UEs, and the fronthaul rate requirements. As a difference from the conference version \cite{AsilomarConf24}, to fully benefit from the centralization in functional split options $8$ and $7.1$, we propose a network power minimization algorithm that is tailored for centralized precoding. The original problem is non-convex, and the SE requirements cannot be written in a closed form in terms of control variables. To tackle this challenge, we converted the original non-convex problem by utilizing scenario sampling approximation and group Lasso optimization methods into an iterative algorithm that solves a convex problem in each iteration. For Option $7.2$, we consider distributed precoding and propose a block coordinate descent-based iterative algorithm to minimize network power consumption. In the numerical analysis, we compare the power consumption under different functional splits and transport technologies. Our results demonstrate that split options $8$ and $7.1$ save significantly more energy compared with Option $7.2$ thanks to the performance improvement by centralized precoding, and sharing the processing resources in the cloud. Furthermore, the proposed joint optimization algorithm improves energy-savings by $15\%$ compared to the radio-only orchestration and by $70\%$ compared to the cloud-only orchestration. 

    \vspace{-4mm}
\subsection{Organization}
The rest of the paper is organized as follows. Section \ref{sec:system_model} describes the considered system model. The signal model and design of the wireless fronthaul are given in Section \ref{sec:wireless_fronth}. Section \ref{sec:power_consumption} describes the calculation of the network power consumption model. 
Section \ref{sec:distributed} and 
Section \ref{sec:centralized} provide network power minimization algorithms considering distributed and centralized precoding, respectively.  Section \ref{sec:numerical_analysis} presents the numerical results, and Section \ref{sec:conclusion} draws the conclusion.

\vspace{-4mm}

\section{System Model}
\label{sec:system_model}
\begin{figure}[t!]
    \centering    \includegraphics[width=0.8\linewidth]{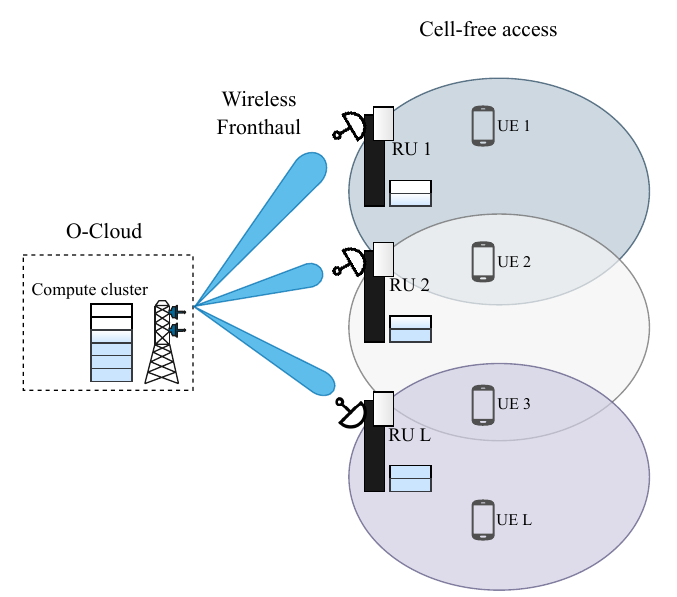}
    \vspace{-3mm}
    \caption{Cell-free massive MIMO architecture illustrating centralized cloud processing with wireless fronthaul.   
    }
    \label{fig:system_model}
    \vspace{-6mm}
\end{figure} 
We consider the downlink of a cell-free massive MIMO system, as illustrated in Fig.~\ref{fig:system_model}, operating in time-division duplex (TDD) mode. We consider that the system is built on top of a specific deployment of O-RAN architecture, in which open distributed unit (O-DU) and open centralized unit (O-CU) are bundled and named as O-Cloud. O-Cloud has the virtualization and resource-sharing capabilities\cite{ozlem_jsac}.  
The system consists of $L$ RUs connected to the O-Cloud via a high-frequency (mmWave) wireless fronthaul link, where RUs and the O-Cloud have line-of-sight (LOS) connectivity.  The RUs serve $K$ single-antenna UEs over a mid-band frequency (sub-6\,GHz) channel. The channel between RUs and the cloud unit will be referred to as the fronthaul channel, while the channel between RUs and UEs will be referred to as the access channel. For the access links, we assume uncorrelated Rayleigh fading channels as in \cite{interdonato2020local}, i.e., the channel from UE $k$ to RU $l$ is $\vect{h}_{kl} \sim \mathcal{N}_{\mathbb{C}}(\vect{0}, \beta_{kl} \vect{I}_{M^{\mathrm{ac}}})$, where $\beta_{kl}>0$ is the respective large-scale fading coefficient.   Each RU is equipped with $M^\mathrm{ac}$ antennas for the access channel and $M^\mathrm{frh}$ antennas for the fronthaul channel. We assume that each antenna element has an individual active transceiver chain, fully capable of digital beamforming.  In the proposed system, we consider that each RU can configure its active transceiver chains based on the quality-of-service (QoS) requirements, fronthaul load limitations, and power minimization. Therefore, we denote the activated antennas at RU $l$  by $M_l \in \{0,\ldots,M^\mathrm{ac}\}$, where $M_l=0$ means that RU is deactivated. The bandwidth utilized in the access channel and the fronthaul channel are denoted as $B^{\mathrm{ac}}$ and $B^{\mathrm{frh}}$, respectively. As can be seen from Fig. \ref{fig:system_model}, the compute cluster is located in the O-Cloud, which is capable of sharing computational resources for all deployed RUs, while each RU also includes compute resources co-located with the radio. Depending on the selected functional split, baseband processing can occur either at the O-Cloud site or at the RU site.

We let $\tau_c$ denote the number of symbols in a TDD frame, where $\tau_p$ symbols for the uplink training signaling, and $\tau_d=\tau_c-\tau_p$ symbols for downlink data. Based on the chosen functional split, the precoding capability of the cell-free massive MIMO system changes. For split options 8 and 7.1, centralized channel estimation and precoding can be utilized, while for split option 7.2, distributed operations must be implemented. Due to space limitations, we omit the explanations for the uplink training phase. 
We assume that during the channel estimation phase, all RUs and antennas are active, as also mandated by the 3GPP protocols \cite{ericsson_2022_energy_5g}. Applying minimum mean-squared error (MMSE) channel estimation, the estimated channel vector between RU $l$ and UE $k$ is denoted by $\hat{\vect{h}}_{kl} \sim \mathcal{N}_{\mathbb{C}}\left(\mathbf{0}, \gamma_{kl}\vect{I}_{M^\mathrm{ac}} \right)$ and the channel estimation error $\tilde{\mathbf{h}}_{k l} \sim \mathcal{N}_{\mathbb{C}}\left(\mathbf{0}, \mathbf{C}_{k l}\right)$, where $\mathbf{C}_{k l} = (\beta_{kl} - \gamma_{kl}) \vect{I}_{M^\mathrm{ac}}$ denote the estimation error correlation matrix. For any antenna element $n$, $\gamma_{kl} = \frac{p_k \tau_{p} \beta_{l, k}^2}{\tau_{p} \sum_{t \in \mathcal{P}_k} p_t \beta_{l, t}+1},$
where $p_k$ is the uplink power of UE $k$, and $\mathcal{P}_k \subset\{1, \ldots, K\}$ is the set of indices that are assigned to the same plot as UE $k$ \cite{interdonato2020local}.

\vspace{-2mm}

\subsection{Downlink Data Transmission with Centralized Precoding (Option 8 and 7.1)}
For the functional split options 7.1 and 8, precoding is done in O-Cloud, corresponding to the centralized precoding schemes in cell-free massive MIMO literature \cite{cell-free-book}. After receiving the precoded downlink signals from the O-Cloud, RUs can simultaneously transmit the same data signal to enable a coherently enhanced signal at each receiving UE. The received downlink signal at UE $k$ is given as 
\begin{equation}
 y_k   =\sum_{i=1}^K \mathbf{h}_k^{\Htran} \mathbf{w}_i \varsigma_i+ n_k,
\end{equation}
where $\varsigma_i$ is the unit-power downlink data signal for UE $i$ ($\mathbb{E}\left\{|\varsigma_{i}|^2\right\}=1$), $\mathbf{h}_k=\left[\mathbf{h}_{k 1}^{\Ttran} \ldots \mathbf{h}_{k L}^{\Ttran}\right]^{\Ttran} \in \mathbb{C}^{LM^{\rm ac}}$ is the collective channel to the UE $k$ from all RUs,  $\mathbf{w}_i=\left[\mathbf{w}_{i 1}^{\Ttran} \ldots \mathbf{w}_{i L}^{\Ttran}\right]^{\Ttran} \in \mathbb{C}^{LM^{\rm ac}}$ is the collective precoding vector intended for UE $i$, and $n_k\sim \CN(0,\sigma^2)$ is the additive noise. Note that, if an RU does not serve a UE, we assume $\mathbf{w}_{i l} = \vect{0}$.

\begin{lemma}
\label{lemma:centralized_SINR}
When UE $k$ knows the average received signal, $\mathbb{E}\left\{\mathbf{h}_k^{\Htran}  \mathbf{w}_k\right\}$, an achievable spectral efficiency (SE) in the downlink operation is
\begin{align}
\mathrm{SE}_k = \frac{\tau_d}{\tau_c} \log_2\left(1 + \mathrm{SINR}_k \right),
\end{align}
where the
    effective signal-to-interference-plus-noise ratio (SINR) of UE $k$, $\mathrm{SINR}_k$ is given in \eqref{eq:SINR_centralized} at the top of the next page. 
\end{lemma}
\begin{figure*}[t!]
\begin{equation}
    \mathrm{SINR}_k = \frac{\left|\mathbb{E}\left\{\hat{\mathbf{h}}_k^{\Htran}  \mathbf{w}_k\right\}\right|^2}{\sum_{i=1}^K \mathbb{E}\left\{ \left|  \hat{\mathbf{h}}_k^{\Htran}  \mathbf{w}_i \right|^2   \right\} + \sum_{i=1}^K  \operatorname{Tr}\left(  \mathbb{E}\left\{   \mathbf{w}_i \mathbf{w}^{\Htran}_i\right\} \tilde{\vect{C}}_k \right) -\left|\mathbb{E}\left\{\hat{\mathbf{h}}_k^{\Htran}  \mathbf{w}_k\right\}\right|^2+\sigma^2}.
    \label{eq:SINR_centralized}
\end{equation}
\hrulefill
\end{figure*}
\begin{proof}
The effective SINR term based on the channel vector is given in \cite[(6.10)]{cell-free-book}. We will reformulate this expression by using the fact that the channel estimation and estimation error are independent and $\hat{\vect{h}}_{k} + \tilde{\mathbf{h}}_{k } = {\mathbf{h}}_{k }$. The numerator term in \cite{cell-free-book} can be reformulated by using the fact that 
    \begin{align}
        \mathbb{E}\left\{\mathbf{h}_k^{\Htran}  \mathbf{w}_k \right\} &=   \mathbb{E}\left\{\hat{\mathbf{h}}_k^{\Htran}  \mathbf{w}_k\right\} + \mathbb{E}\left\{\tilde{\mathbf{h}}_k^{\Htran}  \mathbf{w}_k\right\}  \nonumber  \\
        &\hspace{-12mm}= \mathbb{E}\left\{\hat{\mathbf{h}}_k^{\Htran}  \mathbf{w}_k\right\} + \underset{= {0}}{\underbrace{ \mathbb{E}\left\{\tilde{\mathbf{h}}_k\right\}^{\Htran} \mathbb{E}\left\{ \mathbf{w}_k\right\}}}  = \mathbb{E}\left\{\hat{\mathbf{h}}_k^{\Htran}  \mathbf{w}_k\right\}.
    \end{align} 
The denominator term can be obtained by 
 \begin{align}
        \mathbb{E}\left\{\left|\mathbf{h}_k^{\Htran}  \mathbf{w}_i \right|^2 \right\} & =   \mathbb{E}\left\{  
\mathbf{w}^{\Htran}_i \hat{\mathbf{h}}_k  \hat{\mathbf{h}}_k^{\Htran}   \mathbf{w}_i   \right\} + \mathbb{E}\left\{  \mathbf{w}^{\Htran}_i \tilde{\mathbf{h}}_k  \tilde{\mathbf{h}}_k^{\Htran}   \mathbf{w}_i   \right\}  \nonumber  \\
        &\hspace{-12mm}= \mathbb{E}\left\{ \left|  \hat{\mathbf{h}}_k^{\Htran}  \mathbf{w}_i \right|^2   \right\} + \operatorname{Tr}\left(  \mathbb{E}\left\{   \mathbf{w}_i \mathbf{w}^{\Htran}_i\right\} \tilde{\vect{C}}_k \right),  
    \end{align} 
where $\tilde{\vect{h}}_k \sim \mathcal{N}_{\mathbb{C}}(\boldsymbol{0}, \tilde{\vect{C}}_k)$ and $\tilde{\vect{C}}_k = \mathrm{diag}(\vect{C}_{k1}, \ldots, \vect{C}_{kL})$.   
When these terms are replaced, the effective SINR term can be obtained as in \eqref{eq:SINR_centralized}.
\end{proof}
 Lemma \ref{lemma:centralized_SINR} provides an achievable SE for the downlink operation assuming the precoding vectors are a function of the estimated channel vectors. 
 In the next part, we will provide a closed-form SE specific to distributed precoding.

\vspace{-2mm}

\subsection{Downlink Data Transmission with Distributed Precoding (Option 7.2)}
 For the functional split Option $7.2$, all PHY operations higher than precoding are done in the centralized cloud, where precoding and lower operations are done in the RUs, corresponding to the distributed precoding schemes in cell-free massive MIMO literature \cite{cell-free-book}. After receiving the downlink data signals from the cloud, the RUs can simultaneously apply local precoding and transmit the same data signal to enable a coherently enhanced signal at each receiving UE. The transmit signal from RU $l$ can be given as 
\begin{equation}
    \vect{x}_l = \sum_{k=1}^K \sqrt{\rho_{kl}} \vect{w}_{kl} \varsigma_{k},
\end{equation}
where $\rho_{kl}$ is the  transmit power assigned for UE $k$ at RU $l$. The precoding vector $\vect{w}_{kl} \in \mathbb{C}^{M^{\rm ac}}$ is used by RU $l$ towards UE $k$ and satisfies $\mathbb{E}\left\{\|\vect{w}_{kl} \|^2\right\} = 1$.\footnote{When $M_l$ out of $M^{\rm ac}$  antennas are used, the precoding vector entries corresponding to the idle antennas can be set to zero.} The received signal at UE $k$ becomes 
\begin{equation}
    y_k = \sum_{l=1}^{L} \vect{h}^{\Htran}_{kl} \vect{x}_l + n_k.
\end{equation}
We consider local protective partial zero-forcing (PPZF) as the downlink precoding scheme \cite{interdonato2020local}. 
PPZF is capable of providing the right balance between the array gain and interference cancellation compared to other distributed precoding schemes. In PPZF, each RU divides UEs into two distinct sets: strong-channel UEs and weak-channel UEs. Then, the RUs utilize ZF precoding for the strong-channel UEs, and protective MRT for the weak-channel UEs. The protectiveness of MRT comes from canceling out the interference of weak-channel UEs to strong-channel UEs, creating protection for the strong-channel UEs. We let $\mathcal{S}_l$ and $\mathcal{W}_l$ denote the sets of strong-channel UEs and weak-channel UEs at RU $l$, respectively, where $\mathcal{S}_l \bigcap \mathcal{W}_l = \varnothing$.

\begin{figure*}[tb]
\begin{equation}
    \operatorname{SINR}_k =\frac{\left(\sum_{l=1}^L \sqrt{\left(M_l-\tau_{\mathcal{S}_l}\right) \rho_{k l} \gamma_{k l}}\right)^2}{\sum_{t \in \mathcal{P}_k \backslash\{k\}}\left(\sum_{l=1}^L \sqrt{\left(M_l-\tau_{\mathcal{S}_l}\right) \rho_{tl} \gamma_{kl}}\right)^2+\sum_{t=1}^K \sum_{l=1}^L \rho_{tl}\left(\beta_{kl}-\delta_{kl} \gamma_{kl}\right)+\sigma^2}.
    \label{eq:SINR}
\end{equation}
\hrulefill
\vspace{-4mm}
\end{figure*}

The exact expressions of the precoding vectors are also omitted due to space limitations, but the readers can refer to \cite{interdonato2020local}. An achievable SE for UE $k$ is given by $\mathrm{SE}_k = \frac{\tau_d}{\tau_c} \log_2(1+ \mathrm{SINR}_k)$, where $\mathrm{SINR}_k$ is the effective SINR of UE $k$, and for the considered precoding scheme can be given as in \eqref{eq:SINR} as shown at the top of the page. $\delta_{kl}$ denotes the membership decision of UE $k$, where $ \delta_{kl} = 1$,  if  $k \in \mathcal{S}_l$, or $\delta_{kl} =0$ if  $k \in \mathcal{W}_l$.
$\tau_{\mathcal{S}_l} \leq \tau_{p}$ denotes the number of pilot signals for the strong-channel UEs at RU $l$.
The steps to derive the SINR expression are skipped due to the page limitation; however, the proof can be obtained by following the steps given in Appendix C of \cite{interdonato2020local} by assuming that each RU can have a different number of antennas. 

\vspace{-2mm}
\section{Wireless Fronthaul Design}
\label{sec:wireless_fronth}
We consider a combination of time-division multiple
access (TDMA) and space-division multiple access (SDMA) for the fronthaul channel. Hybrid beamforming is used in the O-Cloud, where it is equipped with $M_c$ antennas driven by $N_c$ RF chains, where $N_c \ll M_c$. O-Cloud divides RUs into distinct groups with a maximum size of  $N_c$. Then, the TDMA protocol is applied between groups. 

\vspace{-2mm}
\subsection{Fronthaul Signal Transmission}
We let $\mathcal{L}_i$ denote the $i$th group of RUs, where $\sum_{i=1}^{\lceil L/N_c \rceil}|\mathcal{L}_i| = L$. The received signal at RU $l$ in $\mathcal{L}_i$ is 
\begin{equation}
    \vect{y}_{l} = \vect{G}_{l} \vect{F}_i \vect{W}_i \vect{s}_i + \vect{n}_{l},
\end{equation}
where $ \vect{G}_{l} \in \mathbb{C}^{   M^{\mathrm{frh}} \times M_c }$ is the downlink fronthaul channel between the cloud and RU $l$. Since the fronthaul links are assumed in LOS, and the RU deployments are static, the fronthaul channel is assumed to be perfectly known in the cloud.  

We let $\vect{s}_i \in \mathbb{C}^{|\mathcal{L}_i|}$ denote the fronthaul signals of RUs in $\mathcal{L}_i$ and $\vect{n}_l\sim \CN(\vect{0},\sigma^2\vect{I}_{M^{\rm frh}})$ is the additive noise. 
$\vect{F}_i \in \mathbb{C}^{M_c \times N_c}$ and $\vect{W}_i\in \mathbb{C}^{N_c \times |\mathcal{L}_i|}$ are the analog and digital precoding matrices. The RUs also perform analog combining, after which the corresponding received signal can be represented as 
\begin{equation}
    {\hat{y}}_{l} = \vect{v}_l^{\Htran} \vect{G}_{l} \vect{F}_i \vect{W}_i \vect{s}_i + \vect{v}_l^{\Htran} \vect{n}_{l},
\end{equation}
where $\vect{v}_l \in \mathbb{C}^{M^{\mathrm{frh}} }$. We assume that O-Cloud chooses the columns of the $\vect{F}_i$ as the array response vectors in the directions of the corresponding RUs. Similarly, the combining vectors are chosen as the array response vectors in the direction from the cloud to the corresponding RU.  
By including the effects of analog beamforming into the channel, we can characterize equivalent channel representation as $(\vect{g}^{\mathrm{eq}}_l)^{\Htran} = \vect{v}_l^{\Htran} \vect{G}_{l} \vect{F}_i $, and $\vect{\Bar{G}}_i = [\vect{g}^{\mathrm{eq}}_1 \ldots \vect{g}^{\mathrm{eq}}_{|\mathcal{L}_i|}]^{\Htran}$. Applying ZF precoding at the cloud, we obtain the achievable data rate of RU $l$ as \cite[Ch. 6]{bjornson2024introduction}

\begin{equation}
R^{\rm frh}_l = t_iB^{\mathrm{frh}}\log_2 \left(1+ \Lambda_{ll} \bar{p}_l \right),
\end{equation}
where $\Lambda_{ll}$ is equal to $1/(\sigma^2\left[(\vect{\Bar{G}}_i\vect{\Bar{G}}_i^{\Htran})^{-1}\right]_{l,l})$ and $t_i$ is the allocated time portion to group $i$, where $\sum_{i=1}^{I} t_i = 1$. $\bar{p}_l$ is the power allocated for the $l$th RU for the fronthaul channel, $\sum_{l=1}^{L} \bar{p}_l \leq P_f$. $I$ is defined as the number of groups, i.e., $I=\lceil L/N_c \rceil$.

\vspace{-4mm}

\subsection{RU Grouping for Fronthaul Access}
In RU grouping, we aim to maximize the orthogonality of the channels in a group to reduce the possible interference.  We utilize the chordal distance as the main metric to model the orthogonality between channels. The chordal distance between fronthaul channels of RU $l$ and RU $l'$  is defined as 
$\zeta_{l,l'} = \frac{\left|(\vect{g}^{\mathrm{eq}}_{l})^{\Htran}\vect{g}^{\mathrm{eq}}_{l'}\right|}{\| \vect{g}^{\mathrm{eq}}_{l} \| \| \vect{g}^{\mathrm{eq}}_{l'} \| },$
where $\zeta_{l,l'} \in [0,1]$. Grouping RUs with lower chordal distance corresponds to grouping RUs with higher orthogonality. As a result, possible interference among RUs is minimized.

One way to group RUs is to minimize the maximum sum chordal distance of a group. We let $\alpha_{l,i} \in \{0,1\}$ to denote the membership of RU $l$ in group $i$, and $\boldsymbol{\alpha}_i = [\alpha_{1,i} \ldots \alpha_{L,i}]^{\Ttran}$. We also concatenate all chordal distances into a matrix, $\boldsymbol{\zeta} \in \mathbb{R}^{L\times L}$, where $[\boldsymbol{\zeta}]_{l,l'} = \zeta_{l,l'}$ for $l\neq l'$ and diagonal entries being zero. 
We define a  binary matrix, $\vect{A}_i = \boldsymbol{\alpha}_i\boldsymbol{\alpha}^{\Ttran}_i \in \left\{0, 1\right\}^{L \times L}$ and $[\vect{A}_i]_{l,l'}= a_{l,l',i}= \alpha_{l,i} \alpha_{l',i}$.  The elements of the matrix can be replaced by the following constraints:
\begin{equation}
\begin{aligned}
&a_{l,l',i} \leq \alpha_{l,i}, \hspace{-2.5mm} \quad a_{l,l',i} \leq \alpha_{l',i}, \hspace{-2.5mm} \quad
 a_{l,l',i} \geq \alpha_{l,i}+\alpha_{l',i}-1,
\label{eq:constraint:binary_matrix}
\end{aligned}
\end{equation} 
where $a_{l,l',i}  \in \{0,1\}$. The optimization problem is given as
\begin{subequations} \label{eq:group_opt1:problem}
\begin{align}
 & \underset{\{\alpha_{l,i}, a_{l,l',i} ,\varsigma \}}{\text{minimize}}  \quad \varsigma \label{eq:group_opt1:objective} \\ & \textrm{subject to} \quad \eqref{eq:constraint:binary_matrix} \quad \nonumber \\ &
 \operatorname{Tr}( \boldsymbol{\zeta}\vect{A}_i)
 \leq \varsigma, \quad \forall i,
 \\ &
 \sum_{i=1}^{I} \alpha_{l,i} = 1, \quad \forall l, \quad \sum_{l=1}^{L} \alpha_{l,i} \leq N_c, \quad \forall i,  \label{eq:group_opt:limitgroupsize} \\ &
 \alpha_{l,i}, a_{l,l',i} \in \{0,1\}, \quad \forall l,l',i \label{eq:group_opt1:binary}.
\end{align}
\end{subequations}

The global optimum can be obtained for this problem by using a branch-and-bound algorithm, which we implemented by MOSEK with CVX in MATLAB.

\paragraph*{Wireless Fronthaul Rate Requirement}
Based on the chosen functional split, the required rate on the fronthaul for an RU  changes. We let $O_{\mathcal{X}}$ denote a coefficient effecting the rate requirement for the fronthaul for chosen split $\mathcal{X} \in \{ 8, 7.1, 7.2\}$ (per antenna for $\mathcal{X}\in \{8,7.1\}$ and per UE for $\mathcal{X}=7.2$). They can be formulated by $O_{8} = 2 f_s N_{\mathrm{bits}}$, $O_{7.1} = 2 T^{-1}_s N_{\mathrm{bits}} N_{\mathrm{used}}$, and $O_{7.2} = 2 T^{-1}_s N_{\mathrm{bits}} N_{\mathrm{used}}$, respectively for options 8, 7.1 and 7.2.

\paragraph*{Fronthaul RF chain configuration} To limit the power consumption in the wireless fronthaul, we utilize a fronthaul RF-chain configuration procedure, where, after the RU activation decision is taken, the O-Cloud checks the groups of the active RUs.  If, within any time group, the number of active RUs is fewer than the available RF chains, the O-Cloud deactivates the unused RF chains until only the maximum number of RUs in a group is equal to the number of active RF chains.

\section{Network Power Consumption Model}
\label{sec:power_consumption}
In this section, we will model network power consumption considering the downlink operation.  
For the considered network architecture in Fig. \ref{fig:system_model}, network power consumption can be calculated as 
\begin{equation}
 P_{\mathrm{tot}} = \sum_{l=1}^{L}  (P_{\mathrm{RU},l} + P^{\mathrm{frth}}_{\mathrm{RU},l})  +    P_{\mathrm{Cloud }} + P^{\mathrm{frth}}_{\mathrm{Cloud}}, 
\end{equation}
where $P_{\mathrm{RU},l} $ is the power consumed at RU $l$, and $P_{\mathrm{Cloud }}$ is the power consumed at the O-Cloud  \cite{ozlem_jsac}. $P^{\mathrm{frth}}_{\mathrm{RU},l}$ and $P^{\mathrm{frth}}_{\mathrm{Cloud}}$ are the power consumed at RU $l$ and O-Cloud for the fronthaul, respectively. The power consumption of the backhaul and the core is ignored in this work, since they have a negligible effect compared to the radio and processing \cite{milano}. 

\vspace{-4mm}
\subsection{Power Consumption at the RU-site}
Power consumption at the RU-site can be categorized under two main factors: $(1)$ transmit and hardware power consumption for the access channel; $(2)$ power consumption for processing done at the RU-site (depends on the chosen functional split). The power consumption of RU $l$  becomes 
\begin{equation}
\begin{aligned}
    P_{\mathrm{RU},l} = & P^{\mathrm{hard}}_{\mathrm{RU},l} + P^{\mathrm{proc}}_{\mathrm{RU},l}.
    \end{aligned}
\end{equation}
In hardware power consumption, we constitute both hardware-dependent static power consumption, and the load-dependent total transmit power: 
\begin{align}
&    P^{\mathrm{hard}}_{\mathrm{RU},l} = M_l P_{\mathrm{st}} \nonumber\\
 &   + \Delta^{\rm tr}\cdot\begin{cases} \sum_{k=1}^K\mathbb{E}\left\{\Vert \vect{w}_k\Vert^2\right\}, & \text{for centralized precoding}, \\\sum_{k=1}^{K} \rho_{kl},  & \text{for distributed precoding},\end{cases}
\end{align}
where $P_{\mathrm{st}}$ is the static power consumption per active RF chain and $\Delta^{\rm tr}\geq 1$ is the slope of the load-dependent transmit power consumption.
$P^{\mathrm{RU}}_{\mathrm{proc},l}$ is the power consumption by the processing done at RU $l$, and it depends on the chosen functional split, and is calculated as  
\begin{equation}
P^{\mathrm{proc}}_{\mathrm{RU},l} = \frac{1}{\sigma^{\mathrm{RU}}_{\mathrm{c}}} \left((1-\mathacr{I}_{8}) P^{\mathrm{proc}}_{\mathrm{RU},0} +   \Delta_{r}\frac{ C_{\mathrm{RU},l} }{ C_{\mathrm{RU}}^\mathrm{max}}\right),
\end{equation}
where $P^{\mathrm{proc}}_{\mathrm{RU},0}$ is the idle processing power, and $C_{\mathrm{RU},l}$ is the giga-operations per second (GOPS) at the RU-site. The value of $C_{\mathrm{RU},l}$ can be calculated by summing the processes given in Table \ref{tab:GOPS_table} marked by RU for chosen functional split. $\mathacr{I}_{\mathcal{X}}$ is a binary variable that is equal to one for split option $\mathcal{X}$, and zero for other split options. $ C_{\mathrm{RU}}^\mathrm{max}$ is the processor efficiency at RU in terms of GOPS/W, which can vary based on the chosen hardware technology. $0 < \sigma^{\mathrm{RU}}_{\mathrm{c}} \leq 1$ is the cooling efficiency at any RU. $\Delta_{r}$ is the slope of the load-dependent part.

 \vspace{-3mm}
\subsection{Power Consumption at the O-Cloud-site}
On the O-Cloud-site, the power consumption is given as
\begin{equation}
\label{eq:power_consumption_cloud}
\begin{aligned}
P_{\mathrm{Cloud }}= & P_{\mathrm{fixed}}  +
\frac{1}{\sigma^{\mathrm{Cloud}}_{\mathrm{c}}}\left( P^{\mathrm{proc}}_{\mathrm{Cloud},0} \sum_{\mathacr{w}=1}^{\mathacr{W}} c_{\mathacr{w}}  +  \Delta_{c}\frac{ C_{\mathrm{Cloud}} }{ C^{\mathrm{max}}_{\mathrm{Cld}} }\right),
\end{aligned}    
\end{equation}
where $P_{\rm fixed}$ is the load-independent fixed power consumption, $P^{\mathrm{proc}}_{\mathrm{Cloud},0} $ is the idle processing power, and $C_{\mathrm{Cloud}}$ is the GOPS at the O-Cloud-site that can be calculated by Table \ref{tab:GOPS_table}. $0 < \sigma^{\mathrm{Cloud}}_{\mathrm{c}} \leq 1$ is the cooling efficiency at O-Cloud. $C^{\mathrm{max}}_{\mathrm{Cld}}$ is the processor efficiency at O-Cloud-site in terms of GOPS/W, where based on the chosen hardware technology. $c_{\mathacr{w}} \in \{0,1\}$, where it is equal to one if the $\mathacr{w}$th processor is used, and zero if it is not required. $\mathacr{W}$ denotes the number of processors in the O-Cloud, and if only radio resources are orchestrated, all processors must have idle powers on. With end-to-end resource allocation, processors can share several RU loads, where $\sum_{\mathacr{w}=1}^{\mathacr{W}} c_{\mathacr{w}} = \lceil\frac{ C_{\mathrm{Cloud}} }{ \mathacr{W}C^{\mathrm{max}}_{\mathrm{Cld}} }\rceil$. $\Delta_{c}$ is the slope of the load-dependent part.

 \vspace{-3mm}
\subsection{Power Consumption of Wireless Fronthaul}
We consider the TDMA/SDMA wireless fronthaul access scheme as described in Section \ref{sec:wireless_fronth}; therefore, we consider a single radio at the O-Cloud site. As in \cite{mmWavehybridpower}, the wireless fronthaul power consumption at the O-Cloud site can be calculated by
\begin{equation}
 P^{\mathrm{frth}}_{\mathrm{Cloud}} = \Delta^{\mathrm{fh}} \sum_{l=1}^{L} \bar{p}_l + M_c P_{\mathrm{PA}}  + N_{c} M_c P_{\mathrm{PS}} +  N_{c} (P_{\mathrm{mix}} + P_{\mathrm{DAC}}), 
\end{equation}
where $P_{\mathrm{PA}}$, $P_{\mathrm{PS}}$, $P_{\mathrm{mix}}$  and $P_{\mathrm{DAC}}$ are the power consumption due to power amplifiers, phase shifters, mixers and digital-to-analog converters (DACs), respectively. $\Delta^{\mathrm{fh}}$ is the slope of the fronthaul transmit power.  The power consumption of a single RU for the fronthaul connectivity is modeled by a constant idle power consumption as in the P2P mmWave receivers given by $P^{\mathrm{fronth}}_{\mathrm{RU},l} = P_{\mathrm{ptp}}$ for all active $l$. If an RU is deactivated, the power consumption of its fronthaul radio will be equal to zero. 

\subsection{Network Power Consumption}
Regardless of the chosen functional split, the network power consumption is influenced by the same set of parameters. Depending on the split, the weights of these parameters change. 
The network power consumption is expressed as\footnote{We assume distributed precoding; the power consumption for centralized precoding is obtained similarly by changing the corresponding transmit power. } 
\begin{align}
\label{eq:power_consumption_e2e}
      & P_{\mathrm{tot}} = \bar{P}_{\mathrm{fixed}} + c_0 \sum_{k=1}^{K} \sum_{l=1}^{L} \rho_{kl} + c_1 \sum_{l=1}^{L}  M_l + c_2 \sum_{l=1}^{L} \mathbb{I}(M_l) \nonumber\\ & + c_3\sum_{l=1}^{L}  \sum_{k=1}^{K} \mathbb{I}(\rho_{kl}) + c_4 \sum_{l=1}^{L} M_l \left( \sum_{k=1}^{K} \mathbb{I}(\rho_{kl}) \right)
      + c_5 \sum_{l=1}^{L} \bar{p}_l, 
\end{align}
where $\mathbb{I}(\cdot)$ is the indicator function, which is equal to one if the input of the function is greater than zero, and equal to zero otherwise. $\bar{P}_{\mathrm{fixed}} = {P}_{\mathrm{fixed}} + M_c P_{\mathrm{PA}}  + N_{c} M_c P_{\mathrm{PS}} +  N_{c} (P_{\mathrm{mix}} + P_{\mathrm{DAC}})$ is the fixed power consumption that is ignored in the optimization problem, and later included in simulation results.

The coefficients can be obtained as $c_0 = \Delta^{\rm tr}$ , $c_5 =  \Delta^{\rm fh}$, $c_1 = P_{\mathrm{st}} + \Xi_c C_{\mathrm{mod},l} + [\Xi_c\mathacr{I}_8 + \Xi_r (1- \mathacr{I}_8)] (C_{\mathrm{filter},l} + C_{\mathrm{DFT},l})$, $c_2 = \Xi_c C_{\mathrm{netw},l}  + P_{\mathrm{ptp}} + (1- \mathacr{I}_8) \frac{1}{\sigma^{\mathrm{RU}}_{\mathrm{c}}} P^{\mathrm{proc}}_{\mathrm{RU},0}$, $c_3 = \Xi_c C_{\mathrm{cod},l} + [\Xi_r \mathacr{I}_{7.2} + \Xi_c (1- \mathacr{I}_{7.2})] C_{\mathrm{map},l}$, $c_4 = [\Xi_r \mathacr{I}_{7.2} + \Xi_c (1- \mathacr{I}_{7.2})] C_{\mathrm{prec},l}$. To calculate these coefficients,  GOPS per unit values are used from the Table \ref{tab:GOPS_table}. Scaled processing efficiencies for O-Cloud and an arbitrary RU can be given as  $\Xi_c = (\Delta_c+\mathacr{W}^{-1}P^{\mathrm{proc}}_{\mathrm{Cloud},0})(\sigma^{\mathrm{Cloud}}_{\mathrm{c}}{C^{\mathrm{max}}_{\mathrm{Cld}})}^{-1}$,  $\Xi_r = \Delta_r(\sigma^{\mathrm{RU}}_{\mathrm{c}}C_{\mathrm{RU}}^{\mathrm{max}})^{-1}$, respectively.

\subsection{GOPS Analysis}
3GGP defines several split options that allow the network to carry out some of the PHY functions in the cloud, while others are in the RUs. In this work, we consider split Option 7.1, 7.2, and 8, where the radio frequency (RF) layer and lower PHY operations are carried out at the RUs, while higher PHY processes such as modulation and coding are carried out at the O-Cloud. The GOPS for the operations considered in this work are given in Table \ref{tab:GOPS_table} \cite{Debaillie2015a,malkowsky2017world,desset2016massive,ozlem_jsac}. $\mathrm{W}_r$ and $\mathrm{SE}_r$ denote the ratio of the bandwidth and the ratio of the SE of a UE for this work to the reference setup \cite{Debaillie2015a}. In the reference setup, $20$\,MHz bandwidth is chosen, and the SE is equal to $6$\,bit/s/Hz. The binary variable $r_{il}$ takes the value of $1$ if RU $l$ serves UE $i$ and zero otherwise. $T_s$ is the OFDM symbol duration, $N_{\rm DFT}$ is the DFT size, $N_{\rm used}$  is the number of used subcarriers, and  $f_s$ is the sampling rate.

\begin{table}[tb]
\centering
\caption{GOPS formulations for various operations and their execution locations under different functional splits.}
\vspace{-2mm}
\resizebox{\linewidth}{!}{%
\begin{tabular}{l|l|l|l|l|l}
Function & GOPS per unit* & Factor & 8 & 7.1 & 7.2 \\ \hline
$C_{\mathrm{filter}, l}$ & ${40 f_s}/{10^9}$ & $M_l$ & O-Cloud & RU & RU \\ 
$ C_{\mathrm{DFT}, l}$ & $\frac{8 N_{\mathrm{DFT}} \log_2(N_{\mathrm{DFT}})}{T_s10^9}$ & $M_l$ & O-Cloud & RU & RU \\ 
$C_{\mathrm{map},l}$ & $ 1.3 \mathrm{W}_r \mathrm{SE}^{1.5}_r $ & $\mathacr{R}_l$ $^{\dagger}$ & O-Cloud & O-Cloud & RU \\
$ C_{\mathrm{prec}, l}$ & $\left(\frac{8 \tau_d N_{\mathrm{used}}}{T_s 10^9 \tau_c} \right)$ & $ M_l \mathacr{R}_l$ & O-Cloud & O-Cloud & RU \\    
$C_{\mathrm{mod},l}$ & $ 1.3 \mathrm{W}_r $ & $M_l$ & O-Cloud & O-Cloud & O-Cloud \\
$C_{\mathrm{cod},l}$ & $ 5.2 \mathrm{W}_r \mathrm{SE}_r $ & $ \mathacr{R}_l$ & O-Cloud & O-Cloud & O-Cloud \\
$C_{\mathrm{netw},l}$ & $ 8 \mathrm{W}_r \mathrm{SE}_r $ & $1$ & O-Cloud & O-Cloud & O-Cloud \\ 
\end{tabular}%
}
\footnotesize{*Total GOPS calculated by multiplying GOPS per unit and unit factor.  \\
$^{\dagger}$ $\mathacr{R}_l = \sum_{i=1}^K r_{il}$ is defined for brevity. }
\label{tab:GOPS_table}
\vspace{-6mm}
\end{table}

\vspace{-2mm}
\section{Network Power Minimization for Distributed Precoding}
\label{sec:distributed}
In this section, we will propose an algorithm that minimizes network power consumption considering split option 7.2 and the distributed precoding operation. In this case, the problem can be formulated as
\vspace{-1mm}
{ \begin{subequations} \label{eq:power_min_opt0:problem}
\begin{align}
 & \underset{\{M_l, \rho_{kl}, \bar{p}_l, t_i \}}{\text{minimize}} \quad  c_0 \sum_{k=1}^{K} \sum_{l=1}^{L} \rho_{kl} + c_1 \sum_{l=1}^{L}  M_l + c_2 \sum_{l=1}^{L} \mathbb{I}(M_l) \nonumber\\ & + c_3\sum_{l=1}^{L}  \sum_{k=1}^{K} \mathbb{I}({\rho}_{kl}) + c_4 \sum_{l=1}^{L} M_l \left( \sum_{k=1}^{K} \mathbb{I}({\rho}_{kl}) \right)
      + c_5 \sum_{l=1}^{L} \bar{p}_l \label{eq:power_min_opt0:objective} \\
     & \textrm{subject to} \nonumber \\ &\operatorname{SINR}_k\geq \upsilon_k, \quad \forall k, \label{eq:power_min_opt0:SINR_constraint} \\
 &  t_i B^{\mathrm{frh}} \log_2\left( 1 + {\Lambda_{ll} \bar{p}_l}  \right) \geq    O_{7.2}\sum_{k=1}^{K} \mathbb{I}({\rho}_{kl}), \quad \forall l ,\label{eq:power_min_opt0:wireless_fronthaul_rate} 
 \\&
\sum_{l = 1}^{L} \alpha_{l,i} \bar{p}_l \leq P_f, \quad  \forall i  ,\label{eq:power_min_opt0:wireless_fronthaul_power} 
   \\&
\sum_{i=1}^{I} t_i \leq 1, \label{eq:power_min_opt0:TDMA}   
\\& 
 \sum_{k=1}^{K} \rho_{kl} \leq  P_t, \quad \forall l,\label{eq:power_min_opt0:access_power_limit} \\&
M_l \in \{\tau_{\mathcal{S}_l} +1,\ldots,M^{\mathrm{ac}}\}, \quad \forall l \label{eq:power_min_opt0:integer} .
\end{align}
\end{subequations}}

The objective function, \eqref{eq:power_min_opt0:objective}, is the network total power consumption when the fixed power component is neglected since it does not change with the optimization variables. \eqref{eq:power_min_opt0:SINR_constraint} ensures that the effective SINR at UE $k$ is greater or equal to the threshold value $\upsilon_k$. \eqref{eq:power_min_opt0:wireless_fronthaul_rate} ensures that the fronthaul rate for RU $l$ is higher than or equal to the required fronthaul rate.
\eqref{eq:power_min_opt0:wireless_fronthaul_power} and \eqref{eq:power_min_opt0:access_power_limit} limit the transmit power in the fronthaul and access links, respectively. \eqref{eq:power_min_opt0:TDMA} is the time allocation limit for the TDMA part of the fronthaul channel. \eqref{eq:power_min_opt0:integer} ensures that the number of active antennas at RU $l$, $M_l$ is an integer variable smaller than or equal to the deployed number of antennas at RU $l$. 
\eqref{eq:power_min_opt0:problem} is a non-convex problem with a combinatorial nature due to the integer variables and indicator functions.

To solve this problem, we first relax $M_l$ to a continuous variable  $\hat{M}_l$. For mathematical convenience, we replace  \eqref{eq:power_min_opt0:integer} with $0 \leq \tilde{M}_l \leq  M^{\mathrm{ac}} - \tau_{\mathcal{S}_l}$, where $\tilde{M}_l = \hat{M}_l - \tau_{\mathcal{S}_l}$. Before handling the indicator functions, we reformulate the SINR constraints by utilizing the auxiliary variables $z_{kl} = \sqrt{\tilde{M}_l \rho_{kl}}$, $\vect{z}_k = [ z_{k1}, \ldots, z_{kL}]^{\Ttran}$, and $ \boldsymbol{\bar{\gamma}}_k = [ \sqrt{\gamma_{k1}}, \ldots, \sqrt{\gamma_{kL}}]^{\Ttran}$.  The reformulation of \eqref{eq:power_min_opt0:SINR_constraint} can be given by 
\begin{equation}
    \left \| \begin{bmatrix} &\sqrt{\upsilon_k} \iota_{k1}\boldsymbol{\bar{\gamma}}_k^{\Ttran} \vect{z}_1 \\ & \vdots\\ &\sqrt{\upsilon_k}  \iota_{kK} \boldsymbol{\bar{\gamma}}_k^{\Ttran} \vect{z}_K \\ & \sqrt{\upsilon_k}  (\boldsymbol{\psi}_k \odot \boldsymbol{\bar{\rho}}) \\ &  \sqrt{\upsilon_k\sigma^2} \end{bmatrix} \right \| \leq \boldsymbol{\bar{\gamma}}_k^{\Ttran} \vect{z}_k, \quad \forall k,
    \label{eq:SOC:SINR}
\end{equation}
where $\boldsymbol{\psi}_k = [\psi_{k1}\vect{1}_K^{\Ttran}, \ldots,  \psi_{kL}\vect{1}_K^{\Ttran}]^{\Ttran}$, $\psi_{kl} = \sqrt{\beta_{kl} - \delta_{kl} \gamma_{kl}}$, $\boldsymbol{\bar{\rho}} = [\bar{\rho}_{11}, \ldots, \bar{\rho}_{KL}]^{\Ttran}$, and $\bar{\rho}^2_{kl} = \rho_{kl}$. $\odot$ denotes the Hadamard product. $\iota_{kk'}=1$ if UE $k'$ is in $\mathcal{P}_k \backslash\{k\}$, otherwise it is equal to zero. This is a second-order cone constraint in a convex form. To guarantee the transformation of $z_{kl}$ we introduce 
\begin{equation}
   0  \leq z_{kl} \leq \sqrt{ \tilde{M}_l } \bar{\rho}_{kl}, \quad \forall k,l.
      \label{eq:z_upper_bound}
\end{equation}
We  replace $\mathbb{I}(M_l)$ with a binary variable $m_l\in \{0,1\}$, where $m_l=1$ if $M_l >0$, zero otherwise. Similarly,  $\mathbb{I}(\rho_{kl})$ can be replaced by $r_{kl}\in \{0,1\}, \forall k, l$. The problem becomes
{ 
\begin{subequations} \label{eq:power_min_opt1:problem}
\begin{align}
 & \underset{\{\tilde{M}_l, \bar{\rho}_{kl}, \bar{p}_l, t_i, \vect{z}_k, r_{kl}, m_l \}}{\text{minimize}} \!\!\!\!  c_0 \sum_{l=1}^{L} \sum_{k=1}^{K} \bar{\rho}^2_{kl} + c_1 \sum_{l=1}^{L}  \tilde{M}_l + c_2 \sum_{l=1}^{L} m_l \nonumber\\ & + \sum_{l=1}^{L} \tilde{c}_{3, l}  \sum_{k=1}^{K} r_{kl} + c_4 \sum_{l=1}^{L} \tilde{M}_l  \left( \sum_{k=1}^{K} r_{kl} \right)
      + c_5 \sum_{l=1}^{L} \bar{p}_l \label{eq:power_min_opt1:objective} \\ & \textrm{subject to}  \quad \eqref{eq:power_min_opt0:wireless_fronthaul_power}, \eqref{eq:power_min_opt0:TDMA}, \eqref{eq:SOC:SINR}, \eqref{eq:z_upper_bound} \nonumber\\
      &\bar{\rho}_{kl} \leq  r_{kl} \sqrt{P_t}, \quad  \forall k,l \\
     & B^{\mathrm{frh}} \log_2\left( 1 + {\Lambda_{ll} \bar{p}_l}  \right) \geq   O_{7.2}\left(\frac{\sum_{k=1}^K r^2_{kl}}{ t_i}\right) , \quad \forall l \label{eq:power_min_opt1:wireless_fronthaul_rate}  \\& 
\sum_{k=1}^{K} \bar{\rho}^2_{kl} \leq  P_t ,\quad \forall l ,\label{eq:power_min_opt1:access_power_limit} \\&
m_{l} \leq \tilde{M}_l \leq m_l (M^{\mathrm{ac}} - \tau_{S_l}), \quad \forall l ,\label{eq:power_min_opt1:M_interval} \\&
\sum_{k=1}^K r_{kl} \leq m_l K, \quad \forall l ,\label{eq:power_min_opt1:AP_activation_UE_assoc} \\ &
r_{kl}, m_l \in \{0,1\},  \quad \ \forall k,l, \label{eq:power_min_opt1:binary_variables}
\end{align}
\end{subequations}
}

where $\tilde{c}_{3,l} = c_3 + c_4 \tau_{S_l} $. This problem is non-convex due to $ \tilde{M}_l \left( \sum_{k=1}^{K} r_{kl} \right)$, the binary constraints in \eqref{eq:power_min_opt1:binary_variables}, and the constraint in \eqref{eq:z_upper_bound}, 
 The global optimum for this problem cannot be guaranteed, but an efficient solution can be obtained by adding auxiliary variables that represent the continuous relaxation of the binary variables, separating the binary and continuous variables into different sub-problems, and alternating between these sub-problems.  We first define continuous auxiliary variables, $0 \leq \tilde{r}_{kl}, \tilde{m}_l \leq 1, \forall k,l$, which will replace the binary variables, $r_{kl}, m_l$,  in the constraints. We also define auxiliary variables for power coefficients, $v_{kl}$ and $u_{kl}$, which are used in removing nonconvexities in \eqref{eq:z_upper_bound}. Ideally, we want a final solution to satisfy $\tilde{r}_{kl} = {r}_{kl}$,  $\tilde{m}_{l} = {m}_{l}$, ${u}_{kl} = {\bar{\rho}}^2_{kl}$, ${v}_{kl} = {\bar{\rho}}_{kl}$. Therefore, we add a mean-square-error (MSE) penalty to minimize the error that can be caused by the relaxations.

The first problem with the continuous variables, except the auxiliary variables for power coefficients, can be written as
{ 
\begin{subequations} \label{eq:power_min_opt_subprob1:problem}
\begin{align}
 & \underset{\{ \tilde{M}_l,  \bar{\rho}_{kl}, \bar{p}_l, t_i, \vect{z}_k, \tilde{r}_{kl}, \tilde{m}_l, u_{kl} \}}{\text{minimize}}   c_0 \sum_{l=1}^{L} \sum_{k=1}^{K}  \bar{\rho}^2_{kl} \nonumber\\
 &+ c_1 \sum_{l=1}^{L}  \tilde{M}_l 
   + c_4 \sum_{l=1}^{L} \tilde{M}_l \left( \sum_{k=1}^{K} r_{kl} \right)  + c_5 \sum_{l=1}^{L} \bar{p}_l \nonumber\\ & + \lambda_1 \sum_{l=1}^{L}  \sum_{k=1}^{K} (r_{kl} - \tilde{r}_{kl})^2 +  \lambda_2 \sum_{l=1}^{L} (m_{l} - \tilde{m}_{l})^2\nonumber\\ &  + \lambda_3 \sum_{l=1}^{L} \sum_{k=1}^{K} (u_{kl} - v^2_{kl})^2 + \lambda_4 \sum_{l=1}^{L} \sum_{k=1}^{K} (\bar{\rho}_{kl} - v_{kl})^2\label{eq:power_min_subprob1:objective} \\
 & \textrm{subject to} \quad  \eqref{eq:power_min_opt0:wireless_fronthaul_power}, \eqref{eq:power_min_opt0:TDMA}, \eqref{eq:SOC:SINR}, \eqref{eq:power_min_opt1:access_power_limit}, \nonumber \\
 & \left\| \left[ \sqrt{2} z_{kl}, \tilde{M}_l, u_{kl} \right] \right\| \leq \tilde{M}_l + 
u_{kl}, \quad \forall k,l,
\label{eq:power_min_subprob1:SOC_antenna_rho} \\
 & B^{\mathrm{frh}} \log_2\left( 1 + {\Lambda_{ll} \bar{p}_l}  \right) \geq   O_{7.2}\left(\frac{\sum_{k=1}^K \tilde{r}^2_{kl}}{ t_i}\right) , \quad \forall l, \label{eq:power_min_subprob1:wireless_fronthaul_rate}  \\&
\tilde{m}_l \leq \tilde{M}_l \leq \tilde{m}_l (M^{\mathrm{ac}} - \tau_{S_l}), \quad \forall l, \label{eq:power_min_subprob1:M_interval} \\   &
\sum_{k=1}^K \tilde{r}_{kl} \leq \tilde{m}_l K, \quad \forall l ,\label{eq:power_min_subprob1:AP_activation_UE_assoc} \\&
u_{kl} \leq  \tilde{r}_{kl} P_t, \quad  \forall l,k, \\ &
0 \leq \tilde{r}_{kl}, \tilde{m}_l \leq 1 ,  \quad \ \forall k,l. \label{eq:power_min_subprob1:binary_variables}
\end{align}
\end{subequations}
}
 For given, $r_{kl}$, $m_l$, and $v_{kl}$ values, \eqref{eq:power_min_opt_subprob1:problem} is in a convex form that can be solved by any convex programming solver. 

The second sub-problem will be solved only to find $v_{kl}$, which can be described by
\begin{subequations} \label{eq:power_min_opt1:problem}
\begin{align}
 & \underset{\{ v_{kl} \}}{\text{minimize}} \,   \lambda_3 \sum_{l=1}^{L} \sum_{k=1}^{K} (u_{kl} - v^2_{kl})^2 + \lambda_4 \sum_{l=1}^{L} \sum_{k=1}^{K} (\bar{\rho}_{kl} - v_{kl})^2.  
      \label{eq:power_min_opt1:objective} 
\end{align}
\end{subequations}
This problem serves two main purposes: (1) $\bar{\rho}^2_{kl} \rightarrow u_{kl}$ with the help of $v_{kl}$, and (2) facilitate solving the first-problem jointly for $\tilde{M}_l$, $\bar{\rho}_{kl}$. The solution for this problem is the positive real root of the following equation:
\begin{equation}
    4\lambda_3 v^3_{kl} - (4\lambda_3 u_{kl} - 2\lambda_4)v_{kl} - 2\lambda_4\bar{\rho}_{kl} = 0.
    \label{eq:subproblem2:solution}
\end{equation}

Finally, the third sub-problem is solved for the binary variables as 
\begin{subequations} \label{eq:subproblem3:problem}
\begin{align}
 & \underset{\{ r_{kl}, m_l \}}{\text{minimize}} \,   c_2 \sum_{l=1}^{L} m_l  + \sum_{l=1}^{L}  \tilde{c}_{3,l}\sum_{k=1}^{K} r_{kl} + c_4 \sum_{l=1}^{L} \tilde{M}_l \left( \sum_{k=1}^{K} r_{kl} \right) \\ & +  \lambda_1 \sum_{l=1}^{L}  \sum_{k=1}^{K} (r_{kl} - \tilde{r}_{kl})^2 +  \lambda_2 \sum_{l=1}^{L} (m_{l} - \tilde{m}_{l})^2
      \label{eq:subproblem3:objective} \\ & \textrm{subject to} \nonumber \\ & 
r_{kl}, m_l \in \{0,1\},  \quad \ \forall l,k. \label{eq:subproblem3:binary_variables}
\end{align}
\end{subequations}
Since only variables in this problem are the binary ones, the optimal solution of this sub-problem can be obtained by just checking the coefficients of these variables:
\begin{subequations} \label{eq:subproblem3:solution}
\begin{align}
& m_{l} = 0.5-0.5\cdot\mathrm{sign}\left(c_2 + \lambda_2 (1- 2\tilde{m}_l)\right)
 , \quad \forall l , \\ & r_{kl} = 0.5-0.5\cdot\mathrm{sign}\left( (\tilde{c}_{3,l} + c_4 \tilde{M}_l) + \lambda_1 (1- 2\tilde{r}_{kl}) \right) \quad \forall k,l. 
\end{align}
\end{subequations}
This approach to separating variables simplifies the problem considerably by eliminating the integer optimization.

  \begin{algorithm}[t]
	 \caption{{Block coordinate descent E2E power minimization for distributed precoding}} \label{alg:block_coordinate_descent}
  	\begin{algorithmic}[1]
                \STATE {\bf Input:}  $c_n$, $n \in {0,1,\ldots,5}$, $\lambda_b$, $b \in {1,\ldots,4}$, $\upsilon_k$, $O_{7.2}$, $a_{l,i}$.
 		\STATE {\bf Initialization:} Initialize $v^{(0)}_{kl}$ and $r^{(0)}_{kl}$ randomly, $m^{(0)}_l = 1$.  Set the iteration counter to $c=0$. Set the approximation accuracy to $\epsilon>0$.
  		\WHILE{ $\mathrm{NMSE} > \epsilon$   }
               \STATE $c \gets c+1$
  		    \STATE Solve \eqref{eq:power_min_opt_subprob1:problem}  with a convex programming solver.
             \STATE Set $ \tilde{M}^{(c)}_l,  \bar{\rho}^{(c)}_{kl}, \bar{p}^{(c)}_l, t^{(c)}_i, \vect{z}^{(c)}_k, \tilde{r}^{(c)}_{kl}, \tilde{m}^{(c)}_l, u^{(c)}_{kl}$ to the solution of \eqref{eq:power_min_opt_subprob1:problem}. 
        \STATE Set $v^{(c)}_{kl}$ to the positive real root of \eqref{eq:subproblem2:solution} for given $\bar{\rho}^{(c)}_{kl}, u^{(c)}_{kl}$.
        \STATE Update  $ r^{(c)}_{kl}, m^{(c)}_l$ based on \eqref{eq:subproblem3:solution}. 
       \STATE Update $\mathrm{NMSE}$. 
	 	\ENDWHILE
  	    \end{algorithmic}
  \end{algorithm}
  
The overall algorithm is described in Algorithm~\ref{alg:block_coordinate_descent}. To ensure a feasible initialization, the algorithm starts by activating all RUs, and with random values of $v_{kl}$. It first solves the sub-problem for continuous variables, \eqref{eq:power_min_opt_subprob1:problem}, then the second sub-problem for  $v_{kl}$, \eqref{eq:subproblem2:solution},  and finally does the binary updates, \eqref{eq:subproblem3:solution}. Then we update the starting point and iterate over all sub-problems until convergence. After convergence, we apply a post-processing procedure to efficiently obtain integer values of $M_l$ from continuous values of $\hat{M}_l=\tilde{M}_l+\tau_{\mathcal{S}_l}$. 

\vspace{-3mm}
\section{Network Power Minimization for Centralized Precoding}
\label{sec:centralized}
In this section, we will propose a network power minimization algorithm considering centralized precoding with split options 8 and 7.1. Due to the centralized precoding, the effective SINR expression is fundamentally different than the distributed precoding scheme as given in \eqref{eq:SINR_centralized}. The lack of a closed-form expression for the effective SINR under centralized precoding prohibits using optimization algorithms with the well-known linear precoding schemes, such as P-MMSE or P-RZF. Therefore, instead of directly injecting the precoding vectors as in the previous section, we will propose a novel approach based on scenario-sampling approximation.

\subsection{Scenario Sampling Approximation}
Scenario sampling is a robust optimization method that is useful when there are probabilistic guarantees in the optimization problem, and the probability distribution function of the random variable is either intractable or expensive to calculate \cite{scenario_sampling}. A relevant version of such a guarantee for our work is the expectation of the optimization variable with respect to a random variable. We let $\boldsymbol{\omega}$ be a random vector having a support $\boldsymbol{\Omega}$, and let the probabilistic constraint be $\mathbb{E}\left\{f(\vect{x}, \boldsymbol{\omega} \right)\} \leq b$, where $f(\mathbf{x}, \boldsymbol{\omega}) :  \mathbf{X} \times \boldsymbol{\Omega} \rightarrow \mathbb{R}$ is a function such that $\mathbb{E}\left\{f(\vect{x}, \boldsymbol{\omega}) \right\}$ is well-defined at all $\vect{x}$. With scenario sampling, $T$ random samples of the random vector $\boldsymbol{\omega}$ is generated, where $\boldsymbol{\omega}_{\vartheta}$ denotes the $\vartheta$th random realization. Then, the expectation is approximated by the sample average, $\mathbb{E}\left\{f(\mathbf{x} , \boldsymbol{\omega}) \right\} \approx \frac{1}{T}\sum_{\vartheta=1}^{T} f(\mathbf{x}, \boldsymbol{\omega}_\vartheta) \leq b$, where the approximation becomes equality as $T\rightarrow \infty$ \cite{scenario_sampling}.  By applying the scenario sampling approach to  \eqref{eq:SINR_centralized}, based on the known channel statistics such as path loss and covariance matrix of the channel estimation error,  we can first create random samples of the estimated channel, $\hat{\mathbf{h}}_k$, and precoding vector, $\mathbf{w}_k$, with ${\hat{\mathbf{h}}}_{k,\vartheta}$, and $\mathbf{w}_{k,\vartheta}$, respectively. The effective SINR can be formulated as in \eqref{eq:SINR_centralized_scenario}. 

\textbf{Remark:} \textit{The samples in \eqref{eq:SINR_centralized_scenario} do not represent actual channel estimates or precoders. Since antenna activation/deactivation decisions taken in O-Cloud are based on long-term statistics (assuming the channel distribution is known and remains stable for several seconds), the samples are drawn from this distribution and discarded after the decision is made.  }

\begin{figure*}[t!]
\begin{equation}
    \mathrm{SINR}_k = \frac{T^{-2}\left|\sum_{\vartheta=1}^T(\hat{\mathbf{h}}_{k,\vartheta}^{\Htran} \mathbf{w}_{k,\vartheta}) \right|^2}{T^{-1} \sum_{i=1}^K \sum_{\vartheta=1}^T \left|  \hat{\mathbf{h}}_{k,\vartheta}^{\Htran}  \mathbf{w}_{i,\vartheta} \right|^2   + T^{-1}\sum_{i=1}^K  \sum_{\vartheta=1}^T     \mathbf{w}^{\Htran}_{i,\vartheta} \tilde{\vect{C}}_k \mathbf{w}_{i,\vartheta} -T^{-2}\left|\sum_{\vartheta=1}^T(\hat{\mathbf{h}}_{k,\vartheta}^{\Htran}  \mathbf{w}_{k,\vartheta})\right|^2+\sigma^2}.
    \label{eq:SINR_centralized_scenario}
\end{equation}
\hrulefill
\end{figure*}
The effective SINR is still in terms of precoding vectors, not as a function of the number of antennas or the power allocation coefficients. Therefore, we will reformulate our original problem as if it were a precoding optimization problem. In this case, $\vect{W} \in \mathbb{C}^{T\times K \times M^{\mathrm{ac}} \times L}$ is a four-dimensional precoding array, denoting the precoding vectors for all RUs, UEs, and at all samples. For example, the precoding vector for the UE $k$ at sample $\vartheta$ from RU $l$ is denoted by $\vect{W}_{\vartheta,k,:,l} \in \mathbb{C}^{M^{\mathrm{ac}}}$. The network power minimization problem for centralized precoding can be given as follows:

{ 
\begin{subequations} \label{eq:centralized_power_min_opt0:problem}
\begin{align}
 & \underset{\{\vect{W}, \bar{p}_l, t_i, b_{l,m} \}}{\text{minimize}} \quad \frac{c_0}{T}   \sum_{l=1}^{L} \|\operatorname{vec}(\vect{W}_{:,:,:,l})\|_2^2 + c_1 \sum_{l=1}^{L} \sum_{m=1}^{M^{\rm ac}} b_{l,m}  \nonumber\\ & + c_2 \sum_{l=1}^{L} \mathbb{I}\left(\sum_{m=1}^{M^{\rm ac}} b_{l,m}\right) + c_3\sum_{l=1}^{L}  \sum_{k=1}^{K} \mathbb{I}(\|\vect{W}_{:,k,:,l}\|_F) 
  \nonumber\\ & + c_4 \sum_{l=1}^{L} \sum_{m=1}^{M^{\rm ac}} b_{l,m}  \left( \sum_{k=1}^{K} \mathbb{I}(\|\vect{W}_{:,k,:,l}\|_F) \right)   + c_5 \sum_{l=1}^{L} \bar{p}_l \label{eq:centralized_power_min_opt0:objective} \\
     & \textrm{subject to} \nonumber \\ &\operatorname{SINR}_k\geq \upsilon_k, \quad \forall k, \label{eq:centralized_power_min_opt0:SINR_constraint} \\
 &  t_i B^{\mathrm{frh}} \log_2\left( 1 + {\Lambda_{ll} \bar{p}_l}  \right) \geq    O_{\mathcal{X}} \sum_{m=1}^{M^{\rm ac}} b^2_{l,m} , \quad \forall l ,\label{eq:centralized_power_min_opt0:wireless_fronthaul_rate} 
 \\&
\sum_{l = 1}^{L} \alpha_{l,i} \bar{p}_l \leq P_f, \quad  \forall i,  \label{eq:centralized_power_min_opt0:wireless_fronthaul_power} 
   \\&
\sum_{i=1}^{I} t_i \leq 1, \label{eq:centralized_power_min_opt0:TDMA}   
\\& 
   \| \vect{W}_{:,:,m,l} \|_F^2 \leq  T P_t b_{l,m} \quad \forall l,m,\label{eq:centralized_power_min_opt0:antenna_limit} \\&   \|\operatorname{vec}(\vect{W}_{:,:,:,l})\|_2^2 \leq TP_{t}, \quad \forall l,\label{eq:centralized_power_min_opt0:access_power_limit}  \\&
b_{l,m} \in \{0,1\}, \quad \forall l,m, \label{eq:centralized_power_min_opt0:integer}
\end{align}
\end{subequations} 
}

where \eqref{eq:centralized_power_min_opt0:objective} is the reformulated network power consumption minimization objective function. $b_{l,m}$ denotes a binary variable of the activation of $m$th antenna of the $l$th RU. \eqref{eq:centralized_power_min_opt0:SINR_constraint} is the effective SINR constraint for the UEs, where \eqref{eq:SINR_centralized_scenario} should be plugged in. \eqref{eq:centralized_power_min_opt0:wireless_fronthaul_rate} is the fronthaul rate constraint, where $O_{\mathcal{X}}$ is the scalar rate factor determined by the chosen functional split $\mathcal{X}\in \{ 7.1, 8\}$. \eqref{eq:centralized_power_min_opt0:wireless_fronthaul_power} and \eqref{eq:centralized_power_min_opt0:TDMA} are the fronthaul power and time allocation limits. \eqref{eq:centralized_power_min_opt0:antenna_limit} can be interpreted as a big-M constraint, assigning zero to the precoding vectors for the deactivated antennas. \eqref{eq:centralized_power_min_opt0:access_power_limit} is the transmit power limitation, and \eqref{eq:centralized_power_min_opt0:integer} ensures that the activation variables are binary. \eqref{eq:centralized_power_min_opt0:problem} is non-convex due to the objective function, binary variables, and the non-convex formulation of SINR as given in \eqref{eq:SINR_centralized_scenario}. 
 We will first reformulate the SINR given in \eqref{eq:SINR_centralized_scenario} to obtain a convex form. We denote $\vect{w}_{k,\vartheta} = \operatorname{vec}(\vect{W}_{\vartheta, k, :, :})$, and $\bar{\vect{w}}_k = \left[ \vect{w}^{\Ttran}_{k,1}, \ldots, \vect{w}^{\Ttran}_{k,T} \right]^{\Ttran}$ denote the concatenated precoding vectors. The  vectors $\bar{\vect{h}}_{k}$ are defined as $\bar{\vect{h}}_{k}=\left[\hat{\vect{h}}_{k,1}^{\Ttran},\ldots, \hat{\vect{h}}_{k,T}^{\Ttran}\right]^{\Ttran}$. The SINR constraints in \eqref{eq:centralized_power_min_opt0:SINR_constraint} can be expressed in second-order cone form as in 
 
{ \begin{equation}
     \left \| \begin{bmatrix}  \hat{\mathbf{h}}_{k,1}^{\Htran}  \mathbf{w}_{1,1} \\  \vdots\\  \hat{\mathbf{h}}_{k,T}^{\Htran}  \mathbf{w}_{1,T}  \\  \vdots\\   \hat{\mathbf{h}}_{k,T}^{\Htran}  \mathbf{w}_{K,T} \\   \tilde{\vect{C}}^{1/2}_k \mathbf{w}_{1,1}  \\  \vdots\\   
 \tilde{\vect{C}}^{1/2}_k \mathbf{w}_{K,T} \\  \sqrt{T} \sigma \end{bmatrix} \right \| \leq \sqrt{\frac{\upsilon_k +1}{ \upsilon_k T} }\operatorname{Re}\left(\bar{\vect{h}}^{\Htran}_k  \bar{\vect{w}}_k\right), \quad \forall k.
 \label{eq:SOC_centralized_SINR}
\end{equation}}

Although \eqref{eq:SOC_centralized_SINR} convexifies the SINR constraints through reformulation, the binary variables require relaxation of the original problem, resulting in a loss of global optimality. In the following, we propose an efficient methodology to address these non-convexities. 

\subsection{Group Sparsity-Based Energy-Efficient Precoding Optimization}
The objective function in \eqref{eq:centralized_power_min_opt0:objective} promotes sparsity by minimizing the activated RUs, active antennas at each RU, and RU-UE associations. The remaining terms also ensure reducing the transmit power and limiting the combination of sparse elements. Instead of using many binary variables, this problem can be reformulated using group sparsity methods, which promote sparsity among groups. One effective approach is Group Lasso, where the objective function can be formulated simply as a sum of norm-2 of the groups \cite{Bach2012Sparsity}. In this way, the objective aims to minimize the number of active groups while effectively setting all elements in a deactivated group to zero. To further promote sparsity, we will also utilize iterative $l_1$ minimization \cite[Section~7]{Bach2012Sparsity}.

In the objective function, the terms related to the antenna number and RU activation are dominant compared to the other parameters. Furthermore, they reduce the other terms since UEs have fewer RUs to associate with. We reformulate the objective function by ignoring the cross-terms (terms with $c_4$ coefficient) and RU-UE association (terms with $c_3$ coefficient). We define $\tilde{b}_{l,m} \in [0,1]$ to represent the continuous versions of the binary variables $b_{l,m}$. The continuous problem in an iteration $c$ can be described by 
\begin{subequations} \label{eq:centralized_power_min_opt1:problem}
\begin{align}
 & \underset{\{\vect{W}, \bar{p}_l, t_i, \tilde{b}_{l,m} \}}{\text{minimize}} \quad \frac{c_0}{T}   \sum_{l=1}^{L} \|\operatorname{vec}(\vect{W}_{:,:,:,l})\|_2^2   + c_5 \sum_{l=1}^{L}   \bar{p}_l  \nonumber\\ &  + (c_1 + c_2) \sum_{l=1}^{L} \sum_{m=1}^{M^{\rm ac}}  \eta_{l,m}^{(c)} \tilde{b}_{l,m}    \label{eq:centralized_power_min_opt1:objective}  
 \end{align}
 \begin{align}
      & \textrm{subject to}  \nonumber \\ & \eqref{eq:SOC_centralized_SINR}, \eqref{eq:centralized_power_min_opt0:wireless_fronthaul_power}, \eqref{eq:centralized_power_min_opt0:TDMA}, \eqref{eq:centralized_power_min_opt0:access_power_limit}, \nonumber \\ 
 &   B^{\mathrm{frh}} \log_2\left( 1 + {\Lambda_{ll} \bar{p}_l}  \right) \geq    O_{\mathcal{X}} \frac{\sum_{m=1}^{M^{\rm ac}} \tilde{b}^2_{l,m}}{t_i} , \quad \forall l, \label{eq:centralized_power_min_opt1:wireless_fronthaul_rate} 
 \\&
   \| \vect{W}_{:,:,m,l} \|_F \leq  \sqrt{TP_t} \tilde{b}_{l,m} \quad \forall l,m,\label{eq:centralized_power_min_opt1:antenna_limit} \\& 
   \tilde{b}_{l,m} \leq b^{(c)}_{l,m}, \quad \forall l,m. \label{eq:centralized_power_min_opt1:binary}
\end{align}
\end{subequations}
The group sparsity in \eqref{eq:centralized_power_min_opt1:objective} is obtained by the summation of the $\tilde{b}_{l,m}$ variables with the weights $\eta^{(c)}_{l,m} = \frac{1}{\tilde{b}^{(c-1)}_{l,m} + \varrho}$. $\varrho$ is an arbitrarily small positive number, and the weights promote the activation variables with small values to be equal to zero. In this way, the remaining activation values are also indirectly pushed to get higher values to guarantee SINR constraints. $b^{(c)}_{l,m}$ in \eqref{eq:centralized_power_min_opt1:binary}  is a binary value that results from thresholding the activation solution in the previous iteration as $b^{(c)}_{l,m} = 1- \mathbb{I}(\tilde{b}^{(c-1)}_{l,m} - \varepsilon)$, where $\varepsilon$ is the threshold value. The thresholding in \eqref{eq:centralized_power_min_opt1:binary} enforces small-valued $b^{(c)}_{l,m}$ to be equal to zero, consequently enforces $W_{:,:,m,l}$ to be zero as well through \eqref{eq:centralized_power_min_opt1:antenna_limit}, and preventing incorrect satisfaction of \eqref{eq:SOC_centralized_SINR}. Thresholding also removes the deactivated antennas from the set in the following iteration, enforcing the objective to deactivate more antennas. The overall algorithm for the centralized precoding (split option 8 and 7.1) is given in Algorithm~\ref{alg:centralized}. Note that Algorithm~\ref{alg:centralized} requires a much smaller number of iterations compared to Algorithm~\ref{alg:block_coordinate_descent}, but the problem size in each iteration is much larger due to the sampling approach. The details of the parameter settings are given in Section \ref{sec:numerical_analysis}.

  \begin{algorithm}[tb]
	 \caption{{Iterative $l1$-based sparsity inducing  E2E power minimization for centralized precoding}} \label{alg:centralized}
  	\begin{algorithmic}[1]
                \STATE {\bf Input:} $\beta_{kl},\gamma_{kl}$, $\vect{\tilde{C}}_k, \upsilon_k, \forall k$,  $O_{7.1}$ or $O_{8}$, $T$, $\epsilon$, $\varepsilon$, $\varrho$.
        \STATE {\bf Sampling:} Create $\hat{\vect{h}}_{k,\vartheta}$ for all $\vartheta, k$. 
 		\STATE {\bf Initialization:} Initialize $\tilde{b}^{(0)}_{l,m}= 1$ for all $l,m$.  Set the iteration counter to $c=0$. Set the approximation accuracy to $\epsilon>0$.
  	\WHILE{ $\mathrm{NMSE} > \epsilon$   }
     \STATE $c \gets c+1$
  		    \STATE Solve \eqref{eq:centralized_power_min_opt1:problem}  with a convex programming solver.
             \STATE Set $\eta^{(c)}_{l,m} = (\tilde{b}^{(c-1)}_{l,m} + \varrho)^{-1}$ and   $b^{(c)}_{l,m} = 1- \mathbb{I}(\tilde{b}^{(c-1)}_{l,m} - \varepsilon)$.
           
            \STATE Update $\mathrm{NMSE}$.
        \ENDWHILE    
         \STATE $M_l = \sum_{m=1}^{M^{\rm ac}} b^{(c)}_{l,m}$, for all $l$.
  	    \end{algorithmic}
  \end{algorithm}

\section{Simulation Results}
\label{sec:numerical_analysis}
\begin{table}[tb!]
\caption{Simulation parameters}
\begin{tabular}{|l|l|l|l|}
\hline
$M^{\rm frh}$,  $M_{c}$ & 64, 256 & $N_{\mathrm{bits}}$ & 12 \\ \hline
$f_s$, $B^{\rm ac}$, $B^{\rm frh}$ & $122.88$, $100$, $1000$\,MHz   &  $T_s$ & $35.68\,\mu$s \\ \hline
$P_t$, $P_f$, pilot pow.  & $5$, $20$, $0.5$\,W  & $P_{\mathrm{fixed}}$  & $120$\,W  \\ \hline
$\sigma^{\mathrm{Cloud}}_{\mathrm{c}}$, $\sigma^{\mathrm{RU}}_{\mathrm{c}}$     & $0.9$, $1$  & $\tau_c$, $\tau_p$  &  $260$, $6$  \\ \hline
$C_{\mathrm{Cld}}^{\max }$, $C^{\max }_{\mathrm{RU}}$ & $360$, $180$ GOPS &  $P_{\mathrm{st}}$ & $6.8$\,W \\ \hline
$\Delta^{r}, \Delta^{c}$ & $74$\,W & $P_{\mathrm{OLT}}$ & $20$\,W \\ \hline
$N_{\mathrm{DFT}}$,$N_{\mathrm{used}}$ & 4096, 2667 & $P_{\mathrm{ptp}}$ & $35$\,W  \\ \hline
$P^{\mathrm{proc}}_{\mathrm{RU},0}$, $P^{\mathrm{proc}}_{\mathrm{Cloud},0}$  & $20.8$\,W  &$\Delta^{\mathrm{tr}}$, $\Delta^{\mathrm{fh}}$  & $4$ \\ \hline
$P_{\mathrm{PA}}$, $P_{\mathrm{PS}}$, $P_{\mathrm{mix}}$ & $25$, $75$, $1000$\,mW & $P_{\mathrm{DAC}}$  & $3.8$\,W \\ \hline
\end{tabular}
\label{tab:simulation_params}
\vspace{-2mm}
\end{table}
We consider a square area of size  $1 \times 1~\text{km}^2$ with a grid-type RU deployment, where the O-Cloud is located at the center of the area. In the fronthaul, we consider O-Cloud and RUs are equipped with a uniform circular array (UCA) and uniform linear arrays (ULAs), respectively.   If not specified, we consider $L=16$ and $M^{\rm ac} = 8$. We consider $3$\,GHz and $28$\,GHz carrier frequency for the access and fronthaul links, respectively. While the fronthaul channel is LOS dominant, uncorrelated Rayleigh fading is assumed in the access channel. The shadowing effect in the access channel is modeled as in \cite{ozlem_jsac}.  The SE requirement of UEs is set to $2$ bit/s/Hz. We consider 5G and beyond access channel properties, as given in Table \ref{tab:simulation_params}. The optical fronthaul and processing values are taken from \cite{ozlem_jsac}, while the wireless fronthaul parameters are taken from \cite{mmWavehybridpower}. The UEs are distributed uniformly in the considered area. We run $50$ Monte Carlo simulations and take the average of the performance results.

\vspace{-2mm}
\subsection{Algorithm convergence and sensitivity}
In this part, we explain the convergence and sensitivity properties. The target $\mathrm{NMSE} = 10^{-4}$ is obtained with $55$ iterations for  Algorithm~\ref{alg:block_coordinate_descent} and with  $5$ iterations  for Algorithm \ref{alg:centralized}. We chose $T=25$ as the sample size to approximate the expectation for Algorithm \ref{alg:centralized}, since the average SE difference between $25$ and $1000$ samples is lower than $10^{-2}$. After implementing Algorithm \ref{alg:centralized}, O-Cloud decides which antennas should be on and utilizes a linear precoding scheme of choice instead of directly using the sparse precoding vectors. The lowest SE among UEs using PMMSE precoding is $1.84$\,bit/s/Hz on average among different setups, while using PRZF precoding, one can obtain $1.91$\,bit/s/Hz. Both schemes are very close to the targeted $2$\,bit/s/Hz, highlighting the applicability of the proposed methodology.

\vspace{-3mm}
\subsection{Power consumption vs functional splits}
\begin{figure}
    \centering
 \subfigure[]{ \label{fig:KcomparisonWF}
	\includegraphics[width=\linewidth]{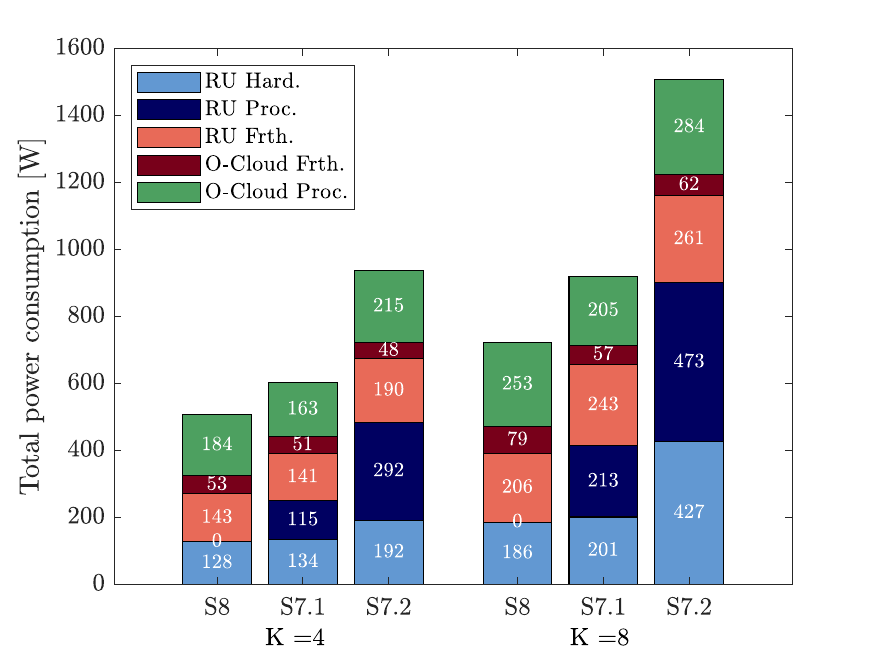} }
		\hspace{-.5cm} 

\subfigure[]{ 
			\label{fig:KcomparisonOF}
	        \vspace{-2mm}
\includegraphics[width=\linewidth]{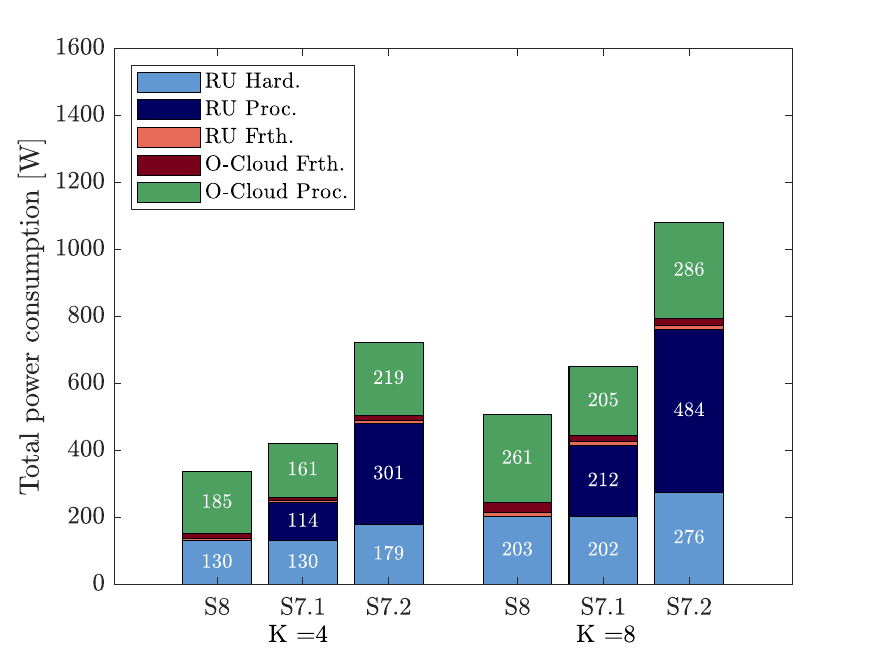} }
    \vspace{-2mm}
    \caption{Total network power consumption for different split options under different UE traffics considering (a) wireless fronthaul, (b) optical fronthaul.}
    \label{fig:Kcomparison}
    \vspace{-3mm}
\end{figure}

Fig. \ref{fig:Kcomparison} compares the detailed network power consumption of different functional split options under different UE densities and different fronthaul transport technologies (Fig. \ref{fig:KcomparisonWF} considers wireless fronthaul, and Fig. \ref{fig:KcomparisonOF} considers optical fronthaul). Below, we explain the power consumption trends for each network component in depth for different functional split options. 

\textit{a) RU hardware:} In Fig. \ref{fig:KcomparisonWF}, while split options 8 and 7.1 are approximately equal in RU hardware power consumption, option 7.2 consumes significantly more power due to the use of distributed precoding. This scheme requires activating a greater number of antennas to meet the same SE requirement as the centralized precoding employed in options 8 and 7.1.

\textit{b) RU processing:} The power consumption of RU processing is mainly influenced by the chosen functional split. While in option 8, all processing is done in the O-Cloud, in option 7.2, all lower-layer processing, including precoding, is done in the RU, naturally increasing the power consumption in the RUs. The RU processing power is also indirectly affected by the UE density, where more UEs result in more RUs with more antennas to be activated. 

\textit{c) RU fronthaul:}  RU fronthaul power mainly depends on the active number of RUs. As shown in the figure, the fronthaul power increases with higher functional split options, indicating that more RUs are active at higher splits. While this trend is expected for option 7.2, options 8 and 7.1 are anticipated to activate a similar number of RUs since both employ centralized precoding. The difference between option 8 and 7.1 arises from the sparseness of the applied algorithm rather than the deployment configuration.  Specifically, stricter fronthaul rate constraint in Algorithm \ref{alg:centralized} promotes a sparser solution for option 8, thereby reducing the number of active RUs. 

\textit{d) O-Cloud fronthaul:} The wireless fronthaul limitation creates an opposite trend in the O-Cloud fronthaul power consumption, where option 8 consumes more power compared to option 7.1, especially under the $K=8$ case. Although a similar number of RUs are activated in both options, due to the higher wireless fronthaul rate requirement in option 8, O-Cloud is required to provide higher multiplexing gains by activating more  RF chains and using more transmit power in the fronthaul. 

\textit{e) O-Cloud processing:} The highest processing power consumption at O-Cloud is expected to be in option 8, since all physical layer operations are carried out in O-Cloud. However, the results demonstrate that option 7.2 results in both higher processing power in the O-Cloud and in the RU-site. Since distributed precoding requires more RUs and more antennas to be activated, the total amount of processing required gets significantly higher, resulting in higher power consumption. 

Similar comments in each component also apply to the optical fronthaul case in Fig. \ref{fig:KcomparisonOF}. Since the optical fronthaul lifts the strict fronthaul rate constraints, options 8 and 7.1 activate a similar number of RUs and antennas, following predictable trends. The total power consumption of the fiber fronthaul is significantly lower than that of the wireless fronthaul, demonstrating the tradeoff between the deployment and operational costs of a network. 

Overall, the results demonstrate significant energy-saving benefits of the functional split option 8 and option 7.1 compared to option 7.2, largely thanks to the performance improvement achieved through centralized precoding. Although centralized precoding requires higher complexity (and thus higher processing power consumption), centralizing the processing in the O-Cloud reduces idle and local processing power at the RU-sites, eventually reducing the total energy consumption even further. While option 7.1 exhibits slightly higher power consumption than option 8, its reduced fronthaul rate requirement makes it a more practical choice for wireless fronthaul scenarios. In contrast, for optical fronthaul deployments where bandwidth constraints are less critical, option 8 emerges as the most ideal functional split.

\subsection{Comparison with benchmark orchestration schemes}

\begin{figure}
    \centering
\label{fig:DeploymentWF}
\includegraphics[width=\linewidth]{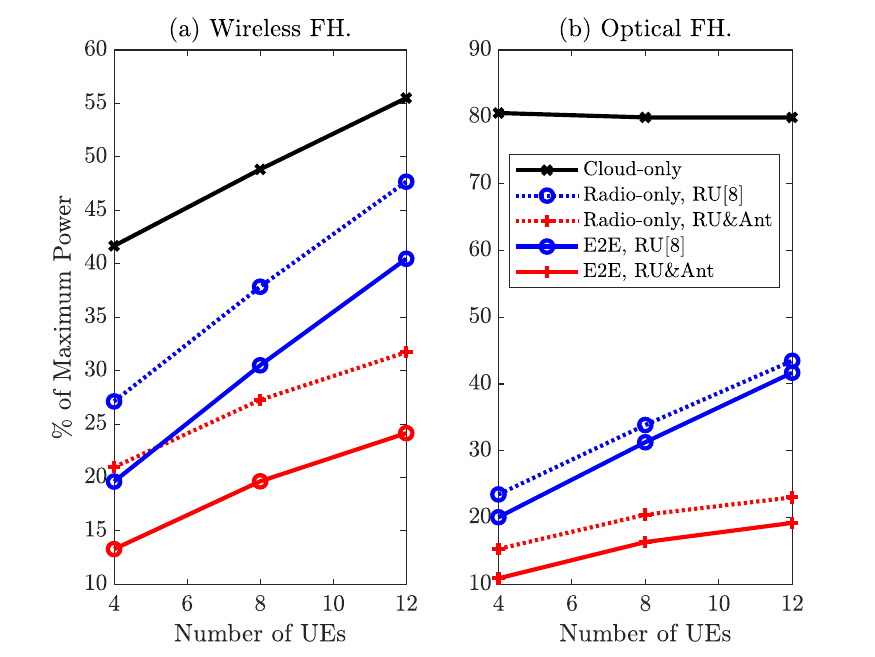} 
\vspace{-4mm}
    \caption{Percentage of power consumption of different algorithms relative to the maximum network power consumption under split 7.1, considering (a) wireless fronthaul, (b) optical fronthaul.  }
    \label{fig:BenchmarkLoad} 
    \vspace{-2mm}
\end{figure}

As shown in Fig. \ref{fig:BenchmarkLoad}, three existing benchmark algorithms and a radio-only resource orchestration version of the proposed algorithm are implemented. All benchmarks target minimizing the considered power consumption, while guaranteeing $2 \text{ bits/s/Hz}$ for the UEs.  Cloud-only orchestration is when the RUs and O-Cloud are unaware of each other. While O-Cloud shares the processing resources as in  \eqref{eq:power_consumption_cloud} in this method, the radio site only minimizes the transmit power. In contrast, in the radio-only orchestration, radio-site power consumption is targeted to be minimized either by shutting down or activating the RUs (as in \cite{ozlem_jsac}), or by configuring each antenna element, while all idle processors and RF-chains in the cloud are active under the wireless fronthaul. E2E orchestration considers cloud and radio processing resources that are jointly orchestrated to minimize the end-to-end power consumption given in \eqref{eq:power_consumption_e2e}. Since the benchmark algorithms in the original works do not consider wireless fronthaul constraints, we adapted these algorithms, transforming them to fit the current problem structure.
  
  As Fig. \ref{fig:BenchmarkLoad} illustrates, the proposed algorithm can reduce the power consumption of the network $87\%$ under the low-load and $76\%$ under the high-load scenarios. Cloud-only orchestration performs the worst compared to the other methodologies, demonstrating the need to deactivate radio resources. Since the fronthaul is limited under the wireless fronthaul, the cloud-only orchestration deactivates unnecessary RUs and fronthaul parts, lowering power consumption more compared to the optical fronthaul scenario. Radio-only orchestration reduces power consumption further $10-20\%$ under wireless fronthaul, and $60\%$ in the optical fronthaul. However, the active RF-chains in O-Cloud for the fronthaul link, and the always-on idle processors, limit the energy-savings compared to end-to-end resource orchestration. As the figure shows, end-to-end orchestration provides $10\%$ further energy-savings, reducing the total power consumption to $13\%$ of the case when all network resources are on. The proposed algorithm provides $10\%$ further energy-savings compared to the RU-shutdown algorithm in \cite{ozlem_jsac}, and scales better when the network load increases thanks to the refined resource adaptability.

\begin{figure}
    \centering
\label{fig:DeploymentWF}
\includegraphics[width=\linewidth]{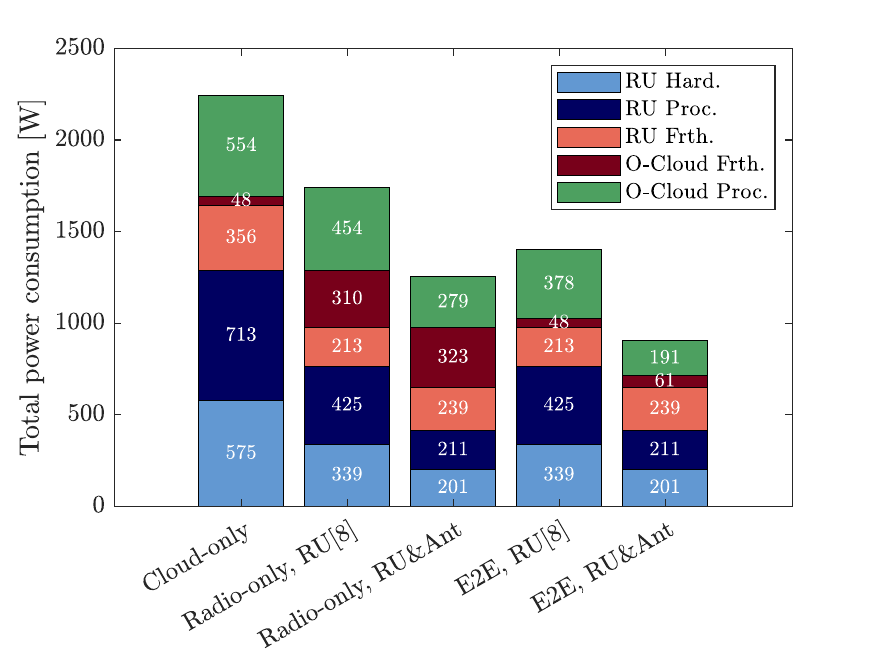} 
    \caption{Power consumption breakdown for different orchestration methods when $K=8$, under wireless fronthaul and functional split option 7.1. }
    \label{fig:BenchmarkBreakdown}
\end{figure}

Fig. \ref{fig:BenchmarkBreakdown} details the power consumption of different network components considering different algorithms when $K=8$, under wireless fronthaul and functional split option 7.1. Cloud-only orchestration consumes more power both in the radio-site and also for the O-Cloud processing. This demonstrates that active radio components not only increase the radio-site power consumption, but also increase the demand in the fronthaul, and create more processing demand, increasing total power consumption at each site. Radio-only orchestration, reduces the power consumption on radio-site by $42\%$ with RU shutdown algorithm, and by $67\%$ with the proposed joint RU and antenna shutdown algorithm. Without deactivating unused RF chains in fronthaul, power consumption grows by $660\%$, demonstrating the importance of the joint orchestration framework. Another interesting observation is that the RU fronthaul power consumption is higher with the proposed algorithm compared to the RU-shutdown algorithm. Since RU fronthaul is mainly affected by the number of active RUs, this trend shows that RU shutdown reduces the number of active RUs compared to the proposed algorithm. However, overall, reducing the total number of antennas provides significantly lower energy consumption due to the reduced hardware and processing resource requirements, where end-to-end orchestration scales all components to minimize their total effect.  


\subsection{The Effect of Deployment}
Fig. \ref{fig:Deployment} compares the power consumption of the network for different deployment densities under different functional splits and fronthaul options. While the x-axis shows the number of deployed RUs, the total number of deployed antennas and total radiated power from RUs in all cases (except $L=36$, where $4$ antennas per RU are deployed) are kept equal. INF cases denote the infeasibility, demonstrating that providing the SE target with the given wireless fronthaul limitation is infeasible with $4$ RUs. This is expected since the distance between UEs and RUs will be much longer when the number of RUs decreases, requiring significantly more antennas to be activated for each RU to compensate for the losses. Eventually, due to the wireless fronthaul limitation, RUs cannot activate the required number of antennas, and the UE rate requirements cannot be satisfied. For Split 7.2, the densest deployment is also infeasible under the wireless fronthaul. Since distributed precoding is used for Split 7.2, the number of UEs associated per RU is limited by the number of antennas deployed. In this case, more RUs need to be activated with all $2$ deployed antennas, eventually creating more infeasibility due to the wireless fronthaul limitation. In all figures, it can be observed that as the deployment becomes denser, RU and cloud fronthaul power consumption increase, especially for the wireless fronthaul case. However, the total  power consumption significantly decreases, regardless of the chosen split or fronthaul type. Deploying denser RUs combined with the proposed end-to-end orchestration mechanism harnesses macro-diversity of cell-free massive MIMO better, consequently reducing hardware, RU, and cloud processing power consumption.

\begin{figure}
    \centering
 \subfigure{ 
			    \label{fig:DeploymentWF}
	\includegraphics[width=\linewidth]{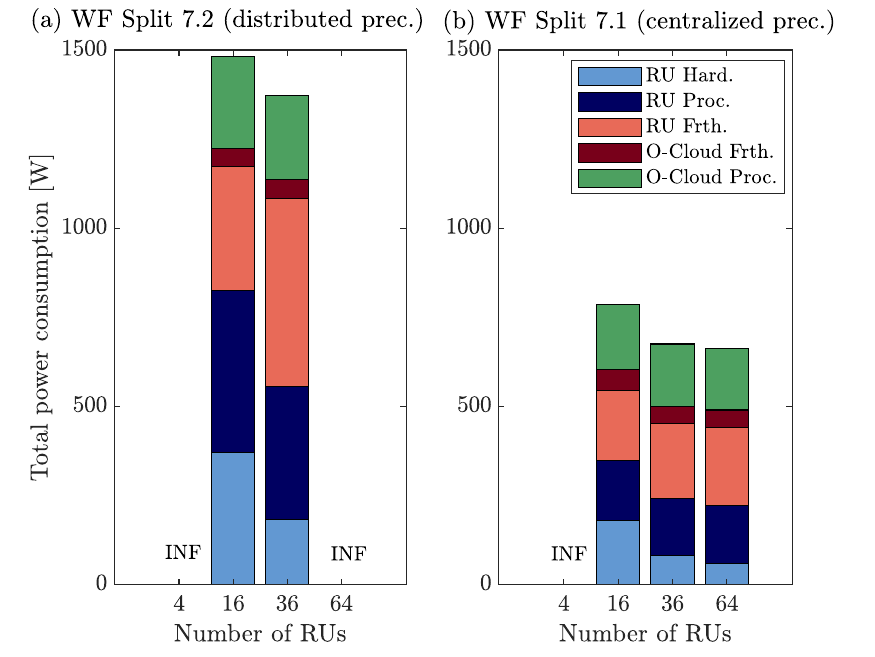} }
		\hspace{-.5cm}
        
        
\subfigure{ 
			\label{fig:DeploymentOF}
	\includegraphics[width=\linewidth]{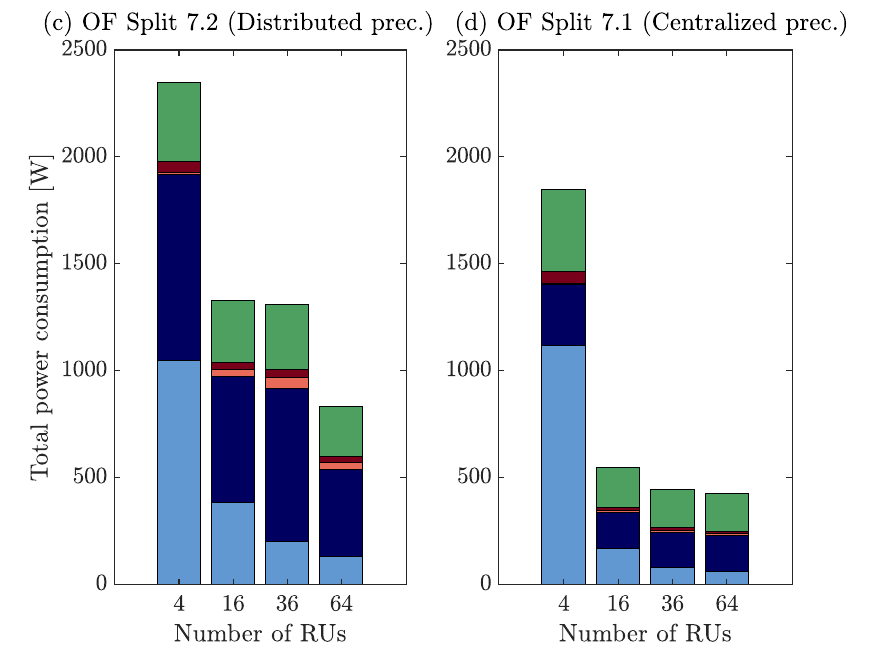} }
    \caption{Total network power consumption for different split options under different numbers of RUs, considering (a) wireless fronthaul (WF) and split 7.2, (b) WF and split 7.1, (c) optical fronthaul (OF) and split 7.2, (d) OF and split 7.1. The total number of antennas is equal in all cases.}
    \label{fig:Deployment}
    \vspace{-2mm}
\end{figure}

\section{Conclusion}
\label{sec:conclusion}
In this work, we investigated energy-efficient cell-free massive MIMO networks with wireless fronthaul through joint antenna, processing, fronthaul, and transmit power optimization.  We have proposed two different power minimization algorithms, one for centralized precoding under split Options 7.1 and 8, and one for distributed precoding for split Option 7.2. We utilized scenario sampling approximation and group sparsity optimization methods to obtain an efficient solution to the original non-convex problem for the centralized precoding. For the distributed precoding, we proposed a block-coordinate descent-based algorithm to efficiently divide the original non-convex mixed integer problem into several blocks of convex problems and closed-form updates. 

Our results demonstrate that, although being computationally more complex, centralized precoding and lower-layer split options provide the lowest network energy consumption by improving the SE performance with less resource requirements. The increased precoding complexity is handled thanks to the shared processing in the cloud, effectively lowering the power consumption by $50\%$  compared to distributed precoding (with Option 7.2). The end-to-end network orchestration outperforms cloud-only and radio-only orchestration by $70\%$ and $15\%$, respectively. Furthermore, by scaling network resources with the number of active antenna elements as proposed, the power consumption can be reduced by $13\%$ compared to scaling with the number of active RUs.  Finally, distributing the same number of antennas across the coverage area and centralizing processing significantly reduces the network power consumption through the proposed antenna activation scheme. These results position cell-free massive MIMO as not only a high-performance architecture but also a compelling, energy-efficient solution for sustainable future networks.

\bibliographystyle{IEEEtran}
\bibliography{IEEEabrv,refs}

\begin{thebibliography}{10}
\providecommand{\url}[1]{#1}
\csname url@samestyle\endcsname
\providecommand{\newblock}{\relax}
\providecommand{\bibinfo}[2]{#2}
\providecommand{\BIBentrySTDinterwordspacing}{\spaceskip=0pt\relax}
\providecommand{\BIBentryALTinterwordstretchfactor}{4}
\providecommand{\BIBentryALTinterwordspacing}{\spaceskip=\fontdimen2\font plus
\BIBentryALTinterwordstretchfactor\fontdimen3\font minus \fontdimen4\font\relax}
\providecommand{\BIBforeignlanguage}[2]{{%
\expandafter\ifx\csname l@#1\endcsname\relax
\typeout{** WARNING: IEEEtran.bst: No hyphenation pattern has been}%
\typeout{** loaded for the language `#1'. Using the pattern for}%
\typeout{** the default language instead.}%
\else
\language=\csname l@#1\endcsname
\fi
#2}}
\providecommand{\BIBdecl}{\relax}
\BIBdecl

\bibitem{cfmMIMOOr}
H.~Q. Ngo, A.~Ashikhmin, H.~Yang, E.~G. Larsson, and T.~L. Marzetta, ``Cell-free massive {MIMO} versus small cells,'' \emph{IEEE Transactions on Wireless Communications}, vol.~16, no.~3, pp. 1834--1850, 2017.

\bibitem{ericsson_2022_energy_5g}
\BIBentryALTinterwordspacing
{Ericsson}, ``Improving energy performance in 5g networks and beyond,'' Ericsson Technology Review, Tech. Rep., 2022, accessed: 2025-07-31. [Online]. Available: \url{https://www.ericsson.com/en/reports-and-papers/ericsson-technology-review/articles/improving-energy-performance-in-5g-networks-and-beyond}
\BIBentrySTDinterwordspacing

\bibitem{GroupSparsePrecGreen}
Y.~Shi, J.~Zhang, and K.~B. Letaief, ``Group sparse beamforming for green {Cloud-RAN},'' \emph{IEEE Transactions on Wireless Communications}, vol.~13, no.~5, pp. 2809--2823, 2014.

\bibitem{EEcfmMIMOSparsePrec}
S.~Chen, J.~Zhang, E.~Bj\"ornson, {\"O}.~T. Demir, and B.~Ai, ``Energy-efficient cell-free massive {MIMO} through sparse large-scale fading processing,'' \emph{IEEE Transactions on Wireless Communications}, vol.~22, no.~12, pp. 9374--9389, 2023.

\bibitem{CellFreeAntennaOptimization}
B.~Yan, Z.~Wang, J.~Zhang, and Y.~Huang, ``Joint antenna activation and power allocation for energy-efficient cell-free massive {MIMO} systems,'' \emph{IEEE Wireless Communications Letters}, vol.~14, no.~1, pp. 243--247, 2025.

\bibitem{minimize_energy_cf}
N.~Jayaweera, K.~B.~S. Manosha, N.~Rajatheva, and M.~Latva-aho, ``Minimizing energy consumption in cell-free massive {MIMO} networks,'' \emph{IEEE Transactions on Vehicular Technology}, vol.~73, no.~9, pp. 13\,263--13\,277, 2024.

\bibitem{radio_min_sparse}
T.~Van~Chien, E.~Björnson, and E.~G. Larsson, ``Joint power allocation and load balancing optimization for energy-efficient cell-free massive {MIMO} networks,'' \emph{IEEE Transactions on Wireless Communications}, vol.~19, no.~10, pp. 6798--6812, 2020.

\bibitem{ozlem_jsac}
{\"O}.~T. Demir, M.~Masoudi, E.~Björnson, and C.~Cavdar, ``Cell-free massive {MIMO} in {O-RAN}: Energy-aware joint orchestration of cloud, fronthaul, and radio resources,'' \emph{IEEE Journal on Selected Areas in Communications}, vol.~42, no.~2, pp. 356--372, 2024.

\bibitem{cloudRAN}
D.~Wang, C.~Zhang, Y.~Du, J.~Zhao, M.~Jiang, and X.~You, ``Implementation of a cloud-based cell-free distributed massive {MIMO} system,'' \emph{IEEE Communications Magazine}, vol.~58, no.~8, pp. 61--67, 2020.

\bibitem{ORAN}
J.~S. Vardakas, K.~Ramantas, E.~Vinogradov, M.~A. Rahman, A.~Girycki, S.~Pollin, S.~Pryor, P.~Chanclou, and C.~Verikoukis, ``Machine learning-based cell-free support in the {O-RAN} architecture: An innovative converged optical-wireless solution toward {6G} networks,'' \emph{IEEE Wireless Communications}, vol.~29, no.~5, pp. 20--26, 2022.

\bibitem{3gpp_tr_38_816}
\BIBentryALTinterwordspacing
``Study on {CU-DU} lower layer split for {NR} (release 18),'' 3rd Generation Partnership Project (3GPP), Technical Report TR 38.816 V18.0.0, 2023, accessed: 2025-07-31. [Online]. Available: \url{https://www.3gpp.org/ftp/Specs/archive/38_series/38.816/38816-f00.zip}
\BIBentrySTDinterwordspacing

\bibitem{larsen_survey_2019}
L.~M.~P. Larsen, A.~Checko, and H.~L. Christiansen, ``A {Survey} of the {Functional} {Splits} {Proposed} for {5G} {Mobile} {Crosshaul} {Networks},'' \emph{IEEE Communications Surveys \& Tutorials}, vol.~21, no.~1, pp. 146--172, 2019, conference Name: IEEE Communications Surveys \& Tutorials.

\bibitem{cell-free-book}
\BIBentryALTinterwordspacing
{\"O}.~T. Demir, E.~Bj\"{o}rnson, and L.~Sanguinetti, ``Foundations of user-centric cell-free massive {MIMO},'' \emph{Foundations and Trends® in Signal Processing}, vol.~14, no. 3-4, pp. 162--472, 2021. [Online]. Available: \url{http://dx.doi.org/10.1561/2000000109}
\BIBentrySTDinterwordspacing

\bibitem{8010338}
W.~Hao and S.~Yang, ``Small cell cluster-based resource allocation for wireless backhaul in two-tier heterogeneous networks with massive {MIMO},'' \emph{IEEE Transactions on Vehicular Technology}, vol.~67, no.~1, pp. 509--523, 2018.

\bibitem{milano}
M.~Brambilla, M.~Cerutti, W.~Colombo, and M.~Tornatore, ``Evaluation of power consumption in {5G} networks at {sub-6 GHz} and {mmWave},'' in \emph{Mediterranean Communication and Computer Networking Conference}, 2023, pp. 43--48.

\bibitem{7306533}
Z.~Gao, L.~Dai, D.~Mi, Z.~Wang, M.~A. Imran, and M.~Z. Shakir, ``{MmWave} massive-{MIMO}-based wireless backhaul for the {5G} ultra-dense network,'' \emph{IEEE Wireless Communications}, vol.~22, no.~5, pp. 13--21, 2015.

\bibitem{ericsson_minilink_6352}
{Ericsson}, ``{Ericsson MINI-LINK 6352 Datasheet},'' \url{https://www.winncom.com/docs/ericsson/Ericsson_MINI-LINK_6352_Datasheet.pdf}, 2022, accessed: 2025-07-31.

\bibitem{umurhan}
U.~Demirhan and A.~Alkhateeb, ``Enabling cell-free massive {MIMO} systems with wireless millimeter wave fronthaul,'' \emph{IEEE Transactions on Wireless Communications}, vol.~21, no.~11, pp. 9482--9496, 2022.

\bibitem{wireless_f}
S.~Elhoushy, M.~Ibrahim, and W.~Hamouda, ``Downlink performance of {CF} massive {MIMO} under wireless-based fronthaul network,'' \emph{IEEE Transactions on Communications}, vol.~71, no.~5, pp. 2632--2653, 2023.

\bibitem{neetu}
N.~R.R., O.~A. Topal, {\"O}.~T. Demir, E.~Björnson, C.~Cavdar, G.~Ghatak, and V.~A. Bohara, ``{UAV}-based cell-free massive {MIMO}: {J}oint activation and power optimization under fronthaul capacity limitations,'' \emph{IEEE Wireless Communications Letters}, pp. 1--1, 2025.

\bibitem{AsilomarConf24}
O.~A. Topal, O.~T. Demir, E.~Bj\"ornson, and C.~Cavdar, ``Energy-efficient cell-free massive {MIMO} with wireless fronthaul,'' in \emph{2024 58th Asilomar Conference on Signals, Systems, and Computers}, 2024, pp. 1591--1596.

\bibitem{interdonato2020local}
G.~Interdonato, M.~Karlsson, E.~Bj{\"o}rnson, and E.~G. Larsson, ``Local partial zero-forcing precoding for cell-free massive {MIMO},'' \emph{IEEE Transactions on Wireless Communications}, vol.~19, no.~7, pp. 4758--4774, 2020.

\bibitem{bjornson2024introduction}
E.~Bj{\"o}rnson and {\"O}.~T. Demir, ``Introduction to multiple antenna communications and reconfigurable surfaces,'' \emph{Now Publishers, Inc.}, 2024.

\bibitem{mmWavehybridpower}
Z.~Hao, Y.~Fang, X.~Yu, J.~Xu, L.~Qiu, L.~Xu, and S.~Cui, ``Energy-efficient hybrid beamforming with dynamic on-off control for integrated sensing, communications, and powering,'' \emph{IEEE Transactions on Communications}, vol.~73, no.~3, pp. 1709--1725, 2025.

\bibitem{Debaillie2015a}
B.~Debaillie, C.~Desset, and F.~Louagie, ``A flexible and future-proof power model for cellular base stations,'' in \emph{VTC Spring}, 2015.

\bibitem{malkowsky2017world}
S.~Malkowsky, J.~Vieira, L.~Liu, P.~Harris, K.~Nieman, N.~Kundargi, I.~C. Wong, F.~Tufvesson, V.~{\"O}wall, and O.~Edfors, ``The world’s first real-time testbed for massive {MIMO}: Design, implementation, and validation,'' \emph{IEEE Access}, vol.~5, pp. 9073--9088, 2017.

\bibitem{desset2016massive}
C.~Desset and B.~Debaillie, ``Massive {MIMO} for energy-efficient communications,'' in \emph{2016 46th European Microwave Conference (EuMC)}.\hskip 1em plus 0.5em minus 0.4em\relax IEEE, 2016, pp. 138--141.

\bibitem{scenario_sampling}
W.~Wang and S.~Ahmed, ``Sample average approximation of expected value constrained stochastic programs,'' \emph{Operations Research Letters}, vol.~36, no.~5, pp. 515--519, 2008.

\bibitem{Bach2012Sparsity}
F.~Bach, R.~Jenatton, J.~Mairal, and G.~Obozinski, \emph{Optimization with Sparsity-Inducing Penalties}.\hskip 1em plus 0.5em minus 0.4em\relax Foundations and Trends in Machine Learning, 2012, vol.~4, no.~1.

\end{thebibliography}

\end{document}